\documentclass[journal]{IEEEtran}
\usepackage[utf8]{inputenc} 
\usepackage[T1]{fontenc}    
\usepackage{hyperref}       
\usepackage{booktabs}       
\usepackage{amsfonts}       
\usepackage{nicefrac}       
\usepackage{microtype}      
\usepackage{xcolor}         
\usepackage{graphicx}
\usepackage{amsthm}
\usepackage{amsmath}
\usepackage{amssymb}
\usepackage{siunitx}
\usepackage{stfloats}
\usepackage{subfigure}
\usepackage{hyperref}
\usepackage{enumerate}
\usepackage{ulem}

\usepackage{xcolor}

\usepackage{ulem}
\usepackage{bm}
\usepackage{nicefrac}
\usepackage{booktabs}
\usepackage{array}
\usepackage{multirow}
\usepackage{threeparttable}
\usepackage{makecell}
\usepackage{caption}

\usepackage{amsmath}

\usepackage{amssymb}
\usepackage{bm}
\usepackage{nicefrac}
\usepackage{booktabs}
\usepackage{array}
\usepackage{multirow}
\usepackage{threeparttable}
\usepackage{makecell}
\usepackage[procnumbered,ruled,vlined,linesnumbered]{algorithm2e}
\usepackage{stfloats}
\usepackage{graphicx}
\usepackage{subfigure}
\usepackage{enumerate}
\usepackage{epstopdf}
\usepackage{tabularx}
\usepackage{stfloats}
\usepackage{booktabs}
\usepackage{siunitx}
\usepackage{hyperref}
\usepackage{balance}

\usepackage{enumitem}

\newtheorem{problem}{Problem}
\newtheorem{theorem}{Theorem}[section]

\newtheorem{lemma}[theorem]{Lemma}

\newtheorem{definition}[theorem]{Definition}
\newtheorem{remark}[theorem]{Remark}

\newenvironment{fminipage}%
{\begin{Sbox}\begin{minipage}}%
		{\end{minipage}\end{Sbox}\fbox{\TheSbox}}

\def\abs#1{\left|#1  \right|}

\def\trace#1{\mathrm{Tr} \left(#1 \right)}
\def\norm#1{\left\| #1 \right\|}
\def\smallnorm#1{\| #1 \|}

\def\calG{\mathcal{G}}
\def\calH{\mathcal{H}}

\def\norm#1{\left\| #1 \right\|}

\def\kh#1{\left( #1 \right)}

\newcommand{\removelatexerror}{\let\@latex@error\@gobble}

\newcommand\LL{\bm{\mathit{L}}}

\def\trace#1{\mathrm{Tr} \left(#1 \right)}

\newcommand\XX{\boldsymbol{\mathit{X}}}
\newcommand\yy{\boldsymbol{\mathit{y}}}

\newcommand\xx{\boldsymbol{\mathit{x}}}
\newcommand\aaa{\boldsymbol{\mathit{a}}}

\newcommand\bb{\boldsymbol{\mathit{b}}}

\newcommand\ee{\boldsymbol{\mathit{e}}}

\renewcommand\AA{\boldsymbol{\mathit{A}}}
\newcommand\BB{\boldsymbol{\mathit{B}}}

\newcommand\JJ{\boldsymbol{\mathit{J}}}
\newcommand\DD{\boldsymbol{\mathit{D}}}

\newcommand\TT{\boldsymbol{\mathit{T}}}

\newcommand\QQ{\boldsymbol{\mathit{Q}}}

\newcommand\vvv{\boldsymbol{\mathit{v}}}
\newcommand{\SDDMSolver}{\textsc{Solve}}

\DeclareMathOperator*{\argmin}{arg\,min}
\DeclareMathOperator*{\argmax}{arg\,max}

\DontPrintSemicolon
\SetKw{KwAnd}{and}
\SetFuncSty{textsc}
\SetKwInOut{Input}{Input\ \ \ \ }

\SetKwInOut{Output}{Output}

\usepackage{tabularx}
\usepackage{stfloats}

\DontPrintSemicolon
\SetKw{KwAnd}{and}
\SetFuncSty{textsc}
\SetKwInOut{Input}{Input\ \ \ \ }
\SetKwInOut{Output}{Output}

\usepackage{tabularx}
\usepackage{stfloats}

\ifCLASSOPTIONcompsoc
  \usepackage[nocompress]{cite}
\else
  \usepackage{cite}
\fi

\ifCLASSINFOpdf

\else

\fi

\begin{document}
\title{Efficient Algorithms for Minimizing the Kirchhoff Index via Adding Edges}

\author{Xiaotian~Zhou, 
        Ahad~N.~Zehmakan,
        and~Zhongzhi~Zhang,~\IEEEmembership{Member,~IEEE}

\IEEEcompsocitemizethanks{
\IEEEcompsocthanksitem This work was supported by the National Natural Science Foundation of China under Grant 62372112 and Grant 61872093. \textit{(Corresponding author: Zhongzhi~Zhang.)} 

\IEEEcompsocthanksitem  Xiaotian~Zhou, and Zhongzhi Zhang  are with Shanghai Key Laboratory of Intelligent Information Processing, School of Computer Science, Fudan University, Shanghai 200433, China.   Zhongzhi Zhang is also with Research Institute of Intelligent Complex Systems, Fudan University, Shanghai 200433, China.
Ahad~N.~Zehmakan is with School of Computing, Australian National University, Canberra, Australia.
\protect\\
E-mail: 22110240080@m.fudan.edu.cn, ahadn.zehmakan@anu.edu.au, zhangzz@fudan.edu.cn
}
\thanks{Manuscript received xxxx; revised xxxx.}
}

\markboth{IEEE Transactions on Knowledge and Data Engineering,~Vol.~xx, No.~xx, April~2024}%
{Shell \MakeLowercase{\textit{et al.}}: Bare Demo of IEEEtran.cls for Computer Society Journals}

\maketitle
\begin{abstract}
The Kirchhoff index, which is the sum of the resistance distance between every pair of nodes in a network, is a key metric for gauging network performance, where lower values signify enhanced performance. In this paper, we study the problem of minimizing the Kirchhoff index by adding edges. We first provide a greedy algorithm for solving this problem and give an analysis of its quality based on the bounds of the submodularity ratio and the curvature. Then, we introduce a gradient-based greedy algorithm as a new paradigm to solve this problem. To accelerate the computation cost, we leverage geometric properties, convex hull approximation, and approximation of the projected coordinate of each point. To further improve this algorithm, we use pre-pruning and fast update techniques, making it particularly suitable for large networks.  Our proposed algorithms have nearly-linear time complexity. We provide extensive experiments on ten real networks to evaluate the quality of our algorithms. The results demonstrate that our proposed algorithms outperform the state-of-the-art methods in terms of efficiency and effectiveness. Moreover, our algorithms are scalable to large graphs with over 5 million nodes and 12 million edges.
\end{abstract}

\begin{IEEEkeywords}
Resistance distance, Kirchhoff index, graph algorithm.
\end{IEEEkeywords}


\IEEEpeerreviewmaketitle


\section{Introduction}

\IEEEPARstart{R}{} egarded as a foundational metric in graph theory, effective resistance—also referred to as resistance distance—is derived in the context of an electrical network. This metric has fostered a plethora of pivotal insights and contributions spanning both theoretical and applied dimensions. On the theoretical side, it has emerged as a linchpin for algorithmic graph theory, facilitating groundbreaking developments in computational algorithms that address key challenges such as spectral graph sparsification~\cite{SpSr08}, maximum flow approximations~\cite{ChKeMaSpTe11}, and the generation of random spanning trees~\cite{MaStTa15}, among others. Concurrently, its applied dimensions have seen successful deployments in many data mining domains such as graph clustering~\cite{AlAnLaOv18, QiHa07}, collaborative recommendation systems~\cite{FoPiReSa07}, 
graph embedding techniques~\cite{CaZhCh18}, image segmentation~\cite{BePaPr08} and network influence mining~\cite{ShYiZh18,LiPeShYiZh19}. 
Additionally, the data management community adopts effective resistance in developing efficient graph systems~\cite{QiDhTaPeWa21} 
and applications~\cite{ShMaWuCh14}. 
This dual relevance—both theoretical and applied—has catalyzed extensive scholarly scrutiny over recent decades~\cite{DoSibu18}. 

Beyond the fundamental concept of resistance distance itself, a large range of graph invariants based on this metric has been defined and explored, most notably the Kirchhoff index~\cite{KlRa93, GhBoSa08}. This index, computed as the aggregate of effective resistances across all pairs of nodes within a graph, has found applications across a diverse spectrum of domains~\cite{ShZh19}. It serves as an instrumental tool for evaluating network connectedness~\cite{TiLe10}, assessing the global utility of social recommender architectures~\cite{WoLiCh16}, and gauging the robustness of first-order consensus algorithms in environments fraught with noise~\cite{PaBa14, QiZhYiLi19,YiZhPa20}. Furthermore, it has been used as the foundation to define efficiency and fairness in information access on social networks~\cite{LiZhZeZh23} and the oversquashing in GNN~\cite{BlWaNaWa23}. Particularly, in the data mining and data management communities, the Kirchhoff index has been used to define an importance measure for edges~\cite{LiZh18, YiShLiZh18}, extract sparse representation of a graphs~\cite{SaHaAbSh21,CuXiZh19}, and measure graph vulnerability and robustness~\cite{FrYaKuToCh23}.  In most setups, a network with smaller values of Kirchhoff index is more beneficial and enhances the performance.


Since the Kirchhoff index encodes the performance of various systems, with smaller Kirchhoff index indicating better performance, a substantial amount of effort has been devoted to optimizing the Kirchhoff index by executing different graph operations, cf.~\cite{BlWaNaWa23}. Motivated by real-world instances such as link recommendation systems in online social platforms, a substantial amount of attention has been devoted to optimizing the Kirchhoff index by adding edges~\cite{SuShLyDo15, PrKoMe22, YaMaQiWe18, HuHeTa19}. The problem studied in this paper falls under the umbrella of this line of work.

We study the following optimization problem: For a given connected undirected unweighted graph $\calG=(V,E)$ with $n$ nodes and $m\ll n^2$ edges, a positive integer $k\ll n$, how to add $k$ nonexistent edges from a candidate edge set $Q=(V \times V) \backslash E$ to graph $\calG$, so that the Kirchhoff index for the resulting graph is minimized. Although this problem has been widely studied by the prior work~\cite{SuShLyDo15, PrKoMe22, YaMaQiWe18, HuHeTa19, BlWaNaWa23},  existing algorithms are not practical for large networks with millions of nodes, which emerge regularly in the real world, since these algorithms have a quadratic (or larger) run time. The main aim of the present work is to address this issue by devising a more efficient algorithm which can handle networks of large size.


The conventional and predominant approach to tackle the Kirchhoff index optimization problem has been the employment of the greedy approach~\cite{SuShLyDo15,YaMaQiWe18}. This algorithm, through its iterative selection of an optimal edge over $k$ rounds, yields a suboptimal solution to the problem. The time complexity of their greedy algorithm stands at $O(kn^3)$, rendering it unsuitable for large-scale networks. Furthermore, we present an approximation ratio bound for the greedy algorithm, underpinned by the bounds of submodularity ratio and the curvature of the related function. Additionally, we enhance the time complexity of the greedy algorithm from $O(kn^3)$ to $O(n^3+kn^2)$ by utilizing Sherman-Morrison formula. This greedy algorithm solves the problem to some extent, but does not fundamentally address the difficulty of solving it on large-scale networks.

Given the substantial computational demands of the greedy algorithm in solving the Kirchhoff index problem, the authors of~\cite{PrKoMe22} introduced several algorithms to expedite the computation. Among these, the most proficient is the state-of-the-art algorithm, \textsc{ColStochJLT}. This algorithm harnesses randomized techniques to augment the quality of the greedy algorithm, incorporating a sub-sampling step and random projection. Its time complexity is denoted as $\tilde{O}(n^2+km)$, where the $\tilde{O}(\cdot)$ notation obscures $\log n$ factors. Nonetheless, even $\tilde{O}(n^2+km)$ can be computationally intensive, and their empirical studies are confined to networks with fewer than 100,000 nodes. Moreover, their sub-sampling step, grounded in empirical observations, lacks rigorous theoretical justification.

To mitigate the algorithm's time complexity, we advocate for a gradient-based greedy algorithm. This gradient-based greedy algorithm finds one optimal edge in each iteration by computing the gradient of each candidate edge, rather  than the marginal decrease of the Kirchhoff index. The primary bottleneck in the traditional greedy algorithm is the necessity to query $|Q|$ instances of marginal decrease in the objective function. This process demands $\Omega(n^2)$ time, especially considering that real-world networks tend to be sparse. Leveraging the gradient-based greedy algorithm, we can substantially curtail the number of queries. This reduction is achieved by harnessing geometric properties to determine the convex hull of $n$ points, the coordinates of which are approximated through projection and the Laplacian solver. This algorithm has nearly-linear time complexity and provides a relative error guarantee in approximating the largest gradient in each iteration. Furthermore, we introduce pre-pruning and rapid update techniques to accelerate this gradient-based greedy algorithm. We also propose an algorithm by determining the convex hull only once, and this algorithm is the most efficient one.  All these gradient-based algorithms reduce time complexities, and we give error analyses for each iteration. Comprehensive experimental evaluations underscore both the efficacy and efficiency of our proposed algorithms.

Our research culminates in the following key contributions:

\textbf{Advanced Analysis of the Kirchhoff Index and Greedy Algorithm:} Upon a detailed examination of the objective function inherent to the Kirchhoff index optimization problem, it becomes evident that the said function is not supermodular in nature. This directly implies that the conventional greedy algorithm, traditionally deployed to address this issue, lacks the desired $(1-1/e)$ approximation guarantee. To mitigate this limitation, we delineate the boundaries of the submodularity ratio, represented as $\gamma$, in conjunction with the curvature $\alpha$ of the function under consideration. As a result, the approximation guarantee of $(1-e^{-\gamma\alpha})/\alpha$ is established for the greedy algorithm. This greedy algorithm has $O(n^3+kn^2)$ time complexity.
In a subsequent advancement, we augment the algorithm's efficiency through the incorporation of two approximation methodologies, culminating in an enhanced greedy algorithm characterized by a $\tilde{O}(k|Q|\epsilon^{-2})$ time complexity, accompanied by a comprehensive error analysis of the marginal decrease.

\textbf{A Novel Paradigm for the Kirchhoff Index Optimization Problem along with Efficient Algorithmic Solutions:} We put forward a gradient-based greedy algorithm tailored for the Kirchhoff index optimization problem. Instead of directly computing the marginal decrease of the Kirchhoff index, this gradient-based greedy algorithm computes the gradient of each candidate edge and finds one edge with the largest gradient. An important property of this approach is the circumvention of the otherwise expensive quadratic number of queries typically encountered in each iteration. This reduction is achieved by leveraging the convex hull approximation algorithm, thereby considerably minimizing the query count. The efficiency of this method is conditional on the fast calculation of each node's coordinates, estimated via random projection in tandem with the Laplacian solver.
This algorithm enjoys an almost linear time complexity and rigorous error analysis. To further improve the computational efficacy, we introduce a novel pre-pruning and update mechanism, which does not affect the algorithm's error analysis.

\textbf{Comprehensive Experimental Evaluation:} To complement our theoretical findings, we conduct a large set of experiments. We evaluate the quality of the algorithms on a diverse set of real-world networks. The outcomes of our experiments demonstrate that our proposed algorithms consistently and significantly outperform the state-of-the-art algorithms. Thus, our proposed solutions not only possess theoretical guarantees, but also are very effective and efficient in practice.

\section{Preliminaries}


\subsection{Notations}

We use $\AA^{\top}$ and $\aaa^{\top}$ to represent the transpose of $\AA$ and $\aaa$, respectively. Let $ \ee_i $ denote the column vector of appropriate dimension, where the $i$-th element is $1$, and other elements are $0$. We define $\mathbf{1}$ to be an appropriate-dimension column vector with all entries being ones. Let $\JJ$ be the matrix with all elements being 1. For a matrix $\AA$, $\AA_{i,j} $ denotes the element of $\AA$ at $i$-th row and $j$-th column. Let $\aaa_i$ be the $i$-th element of the vector $\aaa$. For any vector $\aaa$, we use $\norm{\aaa}_2=\sqrt{\sum_i \aaa_i^2}$ to denote the $\ell_2$ norm of the vector $\aaa$, and use $\norm{\aaa}_{\XX} = \sqrt{\aaa^\top \XX \aaa}$ to denote the matrix norm of the vector $\aaa$ for the given matrix $\XX$.
For two non-negative scalars $a$ and $b$, we use $a \approx_{\epsilon} b$ to denote that $a$ is an $\epsilon$-approximation of $b$ obeying relation $(1-\epsilon) b \leq a \leq(1+\epsilon) b$.


\subsection{Graph and Related Matrices}
We consider $\calG = (V,E)$, a connected, undirected, and unweighted graph comprised of $|V|=n$ nodes and $|E|=m$ edges. Let $N_i$ denote the set of nodes adjacent to a node $i\in V$, in which case the degree of node $i$ is given by $|N_i|$.

The adjacency matrix of graph $\calG$, denoted by $\AA$, is an $n \times n$ matrix such that $\AA_{i,j}=1$ if nodes $i$ and $j$ are adjacent, and $\AA_{i,j}=0$ otherwise. The degree matrix, represented by $\DD$ of graph $\calG$, is a diagonal matrix $\DD=\text{diag}(d_1,d_2,\cdots,d_n)$, where the $i$-th diagonal element $d_i=|N_i|$ represents the degree of node $i$. 

The Laplacian matrix of graph $\calG$, denoted by $\LL$, is expressed as $\LL=\DD-\AA$. This matrix is symmetric, semi-positive definite, and singular, thereby rendering it non-invertible. The eigenvalues of $\LL$ are sorted in non-decreasing order as $0=\lambda_1(\LL)\leq \lambda_2(\LL) \leq \cdots \leq \lambda_n(\LL)$. The pseudoinverse of $\LL$ is denoted by $\LL^\dag$ and calculated as $\LL^\dag = \left(\LL +\frac{1}{n}\JJ\right)^{-1} - \frac{1}{n}\JJ$~\cite{GhBoSa08}.

\subsection{Effective Resistance and Kirchhoff Index}




In any given graph $\calG=(V,E)$, the effective resistance $r_{ij}$ between any two nodes, $i$ and $j$, can be computed  by $r_{ij}=\LL^{\dag}_{i,i}+\LL^{\dag}_{j,j}-\LL^{\dag}_{i,j}-\LL^{\dag}_{j,i}$.

The Kirchhoff index of a graph $\mathcal{G}$ is defined as $K(\mathcal{G})=\frac{1}{2}\sum_{i,j\in V}r_{ij}$.
This index can be computed based on the sum of the inverse non-zero Laplacian eigenvalues~\cite{KlRa93}, or equivalently, the trace of the pseudoinverse of the Laplacian matrix~\cite{GhBoSa08}, as expressed by $K(\mathcal{G}) = n \sum_{i=2}^n \frac{1}{\lambda_i(\LL)} = n \trace{\LL^\dag}$.

\subsection{Concepts Related to Set Function}

A significant portion of set functions in optimization problems are not submodular (or supermodular).
For such functions, one can define some quantities to observe the gap between them and submodular (or supermodular) functions.
For a set $Q$, we use $2^{Q}$ to denote the set of all possible subsets of $Q$.

\begin{definition}[Submodularity ratio~\cite{AbDa11}]
For a non-negative set function $f:2^{Q}\rightarrow \mathbb{R}$, its submodularity ratio is defined as the largest scalar $\gamma$ satisfying
\begin{equation*}
\sum_{u \in T\backslash H} (f(H \cup \{u\})-f(H)) \geq \gamma (f(T)-f(H)), \forall H\subseteq T \subseteq Q.
\end{equation*}
\end{definition}

\begin{definition}[Curvature~\cite{BiBuKrTs17}]
For a non-negative set function $f:2^{Q}\rightarrow \mathbb{R}$, its curvature is defined as the smallest scalar $\alpha$ satisfying
\begin{align*}
f(T)-f(T\backslash \{u\}) \geq (1-\alpha)(f(H)-f(H\backslash \{u\})), \\
\forall H\subseteq T\subseteq Q, u \in H.
\end{align*}
\end{definition}



\section{Problem Formulation}

The Kirchhoff index $K(\calG)$ of a graph $\calG$ is a good measure in many application scenarios. For example, it serves as a measure of overall network connectivity~\cite{TiLe10}, edge centrality within complex networks~\cite{LiZh18}, and robustness of the first-order consensus algorithm in noisy networks~\cite{PaBa14, QiZhYiLi19,YiZhPa20}.
In these practical aspects, a smaller value of $K(\calG)$ indicates that the systems have a better performance. Based on Rayleigh's monotonicity law, adding edges to graph $\calG$ will lead to a decrease of $K(\calG)$~\cite{ElSpVaJaKo11}, which motivates us to study the problem of how to minimize $K(\calG)$ by adding $k$ new edges.

\subsection{Problem Statement}

Based on Rayleigh's monotonicity law, adding a set of edges to the graph will result in a decrease in the effective resistance between any pair of nodes~\cite{ElSpVaJaKo11}. Consequently, the Kirchhoff index of the graph will decrease when new edges are incorporated into the graph. We focus on determining how to optimally add a batch of $k$ new nonexistent edges such that the decrease in the Kirchhoff index is maximized. We present the following Kirchhoff index minimization problem, which has been studied in previous works~\cite{SuShLyDo15,YaMaQiWe18,PrKoMe22,KoAc23,PiSo22,HuHeTa19,BlWaNaWa23}.

\begin{problem}\label{pro:kim}
Given a connected graph $\calG=(V,E)$, a candidate edge set $Q=(V\times V) \backslash E$ and an integer $k\ll |Q|$, we aim to identify a subset $T\subset Q$ of $k$ edges. These edges are then added to the graph to create a new graph $\calG(T)=(V,E\cup T)$ such that the Kirchhoff index $K(T):=K(\calG(T))$ of the graph $\calG(T)$ is minimized. More formally, the objective is to find:
\begin{equation*}
T^* = \argmin_{T \subseteq Q,|T|=k} K(T).
\end{equation*}
\end{problem}

When $\calG$ is a dense graph, the Kirchhoff index $K(\calG)$ is very small for practical applications.  In what follows, we focus on sparse graphs, indicating $m \ll n^2$, which is satisfied in most real-world networks. Also, due to the cost constraints of the problem in real scenarios, we can only add a small number of edges, so we let $k \ll |Q|$ in Problem~\ref{pro:kim}.

The Problem~\ref{pro:kim} is NP-hard~\cite{KoAc23}, and is proved by the reduction from 3-colorability problem. The combinatorial nature of this problem allows for a brute-force approach, which involves exhausting all $\binom{|Q|}{k}$ possible subsets of edges. For each subset of edges $T$, we can compute the Kirchhoff index for its associated graph by inverting a matrix, a process which requires $O(n^3)$ time. This results in a total complexity of $O(\binom{|Q|}{k}n^3)$, which suffers from a combinatorial explosion as $k$ increases, making it computationally infeasible even for moderately sized $k$. To avoid such high computational cost, we will explore the way to solve this problem in next sections.

Note that although the Kirchhoff index minimization problem has been previously proposed and studied on unweighted graphs, both the problem and our solutions can be extended to weighted graphs if the candidate edges have the same weight.

\subsection{Deterministic Greedy Algorithm}
In~\cite{SuShLyDo15,PrKoMe22}, the authors proposed a greedy algorithm by adding edges with the largest marginal decrease (best decrease of the Kirchhoff index) in each of the $k$ iterations. To bypass the repeated operations of marginal decrease computation involved in matrix inversion, they adopted the Sherman-Morrison formula~\cite{Me73} to execute rank-1 updates of matrices. To be specific, for any edge $e=(i,j)\in Q$, the marginal decrease $\Delta(e)=K(\calG)-K(\{e\})$ is computed by the Sherman-Morrison formula as
\begin{equation}\label{equ:sm1}
    \Delta(e) =  n \trace{\LL^\dag} - n \trace{(\LL+\bb_e\bb_e^\top)^\dag} = n\frac{\bb_e^\top \LL^{2\dag} \bb_e}{1+\bb_e^\top \LL^\dag \bb_e},
\end{equation}
where $\bb_e=\ee_i-\ee_j$.
Moreover, after adding an edge $e\in Q$ to the graph, the pseudoinverse of the Laplacian matrix of the new graph can be updated by
\begin{equation}\label{equ:sm2}
    (\LL+\bb_e\bb_e^\top)^\dag = \LL^\dag - \frac{\LL^{\dag}\bb_e \bb_e^\top \LL^{\dag}}{1+\bb_e^\top \LL^\dag \bb_e}.
\end{equation}
Based on Equations~\eqref{equ:sm1} and~\eqref{equ:sm2}, in~\cite{SuShLyDo15,PrKoMe22} a simple greedy algorithm was designed as outlined in Algorithm~\ref{alg:det} excluding Line 2 and Line 10. This simple greedy first computes the pseudoinverse of $\LL$ as a preprocessing step in $O(n^3)$ (Line 1).
Then, in each iteration, the marginal gain $\Delta(e)$ of nonexistent edges $e\in Q$ are computed in $O(n)$ time per edge (Line 4). The edge with the largest marginal decrease is added to the graph (Line 5), and the pseudoinverse is updated in $O(n^2)$ (Line 9). Thus, the total time complexity of this greedy algorithm is in $O(n^3+k|Q|n)$~\cite{SuShLyDo15,PrKoMe22}.

However, their proposed greedy algorithm can be improved. Next, we provide a more efficient greedy algorithm, and analyze its quality.

\subsubsection{Reducing Time Complexity.}
We notice that by using the Sherman-Morrison formula, the simple greedy algorithm maintains the pseudoinverse of the Laplacian matrix, which results in the $O(n)$ time cost for computing the marginal decrease $\Delta(e)$ for each $e\in Q$. However, it is still time-consuming. To accelerate the marginal decrease computation process, we maintain two matrices, $\LL^\dag$ and $\LL^{2\dag}$, during $k$ iterations, which is different from the greedy algorithm proposed in~\cite{SuShLyDo15,PrKoMe22}.  Our proposed deterministic greedy algorithm is called \textsc{Deter}, outlined in Algorithm~\ref{alg:det}.

First, we precompute two matrices, $\LL^\dag$ and $\LL^{2\dag}$, in Lines 1 and 2, in $O(n^3)$ time. Then, in each iteration, the computation of $\Delta(e)$ can be reduced to $O(1)$ time complexity per edge. After we select one optimal edge $e$, we add it to the graph and update $\LL^\dag$ and $\LL^{2\dag}$. According to the Sherman-Morrison formula, we have
\begin{align*}
   & \left(\LL+\bb_e\bb_e^\top\right)^{2\dag} = \left(\LL^\dag - \frac{\LL^\dag \bb_e\bb_e^\top\LL^\dag}{1+\bb_e^\top\LL^\dag \bb_e}\right)^2 \\
    =& \LL^{2\dag} + \frac{\LL^\dag \bb_e \bb_e^\top\LL^{2\dag}\bb_e\bb_e^\top \LL^\dag}{(1+\bb_e^\top\LL^\dag \bb_e)^2} - \frac{\LL^{2\dag}\bb_e\bb_e^\top \LL^\dag}{1+\bb_e^\top\LL^\dag \bb_e} - \frac{\LL^{\dag}\bb_e\bb_e^\top \LL^{2\dag}}{1+\bb_e^\top\LL^\dag \bb_e},
\end{align*}
indicating that $\LL^{2\dag}$ can be updated in $O(n^2)$ time complexity in Line 10.

Based on this improvement, our proposed deterministic greedy algorithm \textsc{Deter} has the total time complexity of $O(n^3+kn^2)$, which is smaller than $O(n^3+k|Q|n)$~\cite{SuShLyDo15,PrKoMe22}, since $Q$ is quadratic in $n$ under the assumption that $m \ll n^2$.

Finally, it should be mentioned that the idea for algorithm \textsc{Deter} was initially proposed in~\cite{BlWaNaWa23}. Here we offer a more detailed analysis and a clearer expression of $(\LL+\bb_e\bb_e^\top)^{2\dag}$, which helps us better understand how the algorithm \textsc{Deter} operates and how its time complexity is reduced.

\begin{small}
\begin{algorithm}[tb]
\begin{small}
	\caption{\textsc{Deter}$(\calG, S_1, S_0, Q, k)$}
    \label{alg:det}
	\Input{
		A connected graph $\calG=(V,E)$;
        a candidate edge set $Q=(V\times V)\backslash E$;
        an integer $1 \leq k \leq |Q|$\\
	}
	\Output{
		A subset $T\subseteq Q$ with $|T|=k$
	}
Compute $\LL^{\dag}$\;
Compute $\LL^{2\dag}$\;
Initialize solution $T = \emptyset$ \;
\For{$i = 1$ to $k$}{
Compute $\Delta(e)= K(\calG)-K(\{e\})$ for each $e \in Q$\;
Select $e_i$ s. t. $e_i \gets \mathrm{arg\, max}_{e \in Q} \Delta(e)$\;
Update solution $T \gets T \cup \{ e_i \}$ \;
Update the graph $\calG \gets \calG(V, E \cup \{ e_i \})$ \;
Update $\LL^{\dag} \gets \LL^{\dag} - \frac{ \LL^\dag\bb_e \bb_e^\top\LL^\dag }{1+\bb_e^\top \LL^\dag \bb_e}$\;
Update $\LL^{2\dag} \gets \LL^{2\dag} + \frac{\LL^\dag \bb_e \bb_e^\top\LL^{2\dag}\bb_e\bb_e^\top \LL^\dag}{(1+\bb_e^\top\LL^\dag \bb_e)^2} - \frac{\LL^{2\dag}\bb_e\bb_e^\top \LL^\dag}{1+\bb_e^\top\LL^\dag \bb_e} - \frac{\LL^{\dag}\bb_e\bb_e^\top \LL^{2\dag}}{1+\bb_e^\top\LL^\dag \bb_e}$\;

Update the candidate edge set $Q \gets Q \backslash  \{e_i\} $ \;}
\Return $T$.
\end{small}
\end{algorithm}
\end{small}

\subsubsection{Evaluating the Greedy Algorithm's Quality}

Although the objective function $K(\cdot)$ of Problem~\ref{pro:kim} is not supermodular~\cite{SuShLyDo17correct}, the greedy algorithm usually has a good quality~\cite{BiBuKrTs17,AbDa11} with a tight approximation guarantee of $(1-e^{-\gamma\alpha})/\alpha$~\cite{BiBuKrTs17} where $\gamma$ and $\alpha$ are submodularity ratio and curvature of the non-negative increasing function $f(T) = K(\calG)- K(T)$ for $T \subseteq Q$, respectively. The effectiveness of Algorithm~\ref{alg:det} is summarized in Theorem~\ref{them:appx}, which provides a relative error guarantee compared to the optimal solution and can be proved similarly using the method in~\cite{BiBuKrTs17}.


\begin{theorem}\label{them:appx}
Let $\gamma \in [0,1]$ and $\alpha \in [0,1]$ be submodularity ratio and curvature of the function $f(T) = K(\calG)- K(T)$ for $T \subseteq Q$. Then, the edge set $T$ returned by Algorithm~\ref{alg:det} satisfies
\begin{equation}
K(\calG)-K(T) \geq \frac{1}{\alpha}(1-e^{-\alpha\gamma})(K(\calG)-K(T^*)),
\end{equation}
where 
\begin{equation*}
T^* =\argmin_{T \subseteq Q, |T|=k} K(T).
\end{equation*}
\end{theorem}

Since the effectiveness of the greedy algorithm in Problem~\ref{pro:kim} can be evaluated by determining its submodularity ratio $\gamma$ and curvature $\alpha$, we next provide bounds for $\gamma$ and $\alpha$, respectively.


\begin{theorem}
\label{lem:sub}
The submodularity ratio $\gamma$ of the set function $f(S)=K(\calG)-K(S)$ is bounded as follows:
\begin{equation}
\label{eq:lowgamma}
1 > \gamma \ge \kh{\frac{ \lambda_2(\LL)}{n}}^2>0,
\end{equation}
and its curvature $\alpha$ is constrained by:
\begin{equation}
\label{eq:upalpha}
0 < \alpha \le 1- \kh{\frac{ \lambda_2(\LL)}{n}}^2< 1.
\end{equation}
\end{theorem}
\begin{proof}
	Let $Q= (V\times V)\backslash E$ be the candidate set and $S,T$ be any two subsets of candidate edge set $Q$. To begin with, we first derive a lower and upper bound for the marginal benefit function $ g_{T}(S)=f(S\cup T)-f(S)$, respectively.

 On the one hand,
	\begin{align*}
		& g_{T}(S)  =f(S\cup T)-f(S)
		= K(S) - K(S\cup T) \\
		=& n \trace{\LL({S})^\dagger} - n\trace{\LL({S\cup T})^\dagger} \\
		=& \sum_{i=2}^{n}\frac{n}{\lambda_i(\LL(S))} - \frac{n}{\lambda_i(\LL(S\cup T))} \\
		 =& n \sum_{i=2}^{n}\frac{\lambda_i(\LL({S\cup T})) - \lambda_i(\LL({S}))}{\lambda_i(\LL({S}))\lambda_i(\LL({S\cup T}))} \\
		\ge& n\frac{\trace{\LL({S\cup T})}-\trace{\LL({S})}}{\lambda_{n}(\LL({S}))\lambda_{n}(\LL({S\cup T}))} \\
		= & \frac{2n|T\backslash S|}{\lambda_{n}(\LL({S}))\lambda_{n}(\LL({S\cup T}))}.
	\end{align*}
	On the other hand,
	\begin{align*}
	& g_{T}(S)  = n \sum_{i=2}^{n}\frac{\lambda_i(\LL(S\cup T)) - \lambda_i(\LL(S))}{\lambda_i(\LL(S))\lambda_i(\LL(S\cup T))} \\
	&\le n\frac{\trace{\LL(S\cup T)}-\trace{\LL(S)}}{\lambda_2(\LL(S))\lambda_2(\LL(S\cup T))}
	 = \frac{2n|T\backslash S|}{\lambda_2(\LL(S))\lambda_2(\LL(S\cup T))}.
	\end{align*}

Then, we put the above two bounds together and derive the lower bound of the submodular ratio $\gamma$.
	\begin{align*}
		&\frac{\sum_{e\in T\backslash S} g_{\{e\}}(S)}{ g_{T}(S)} \\
		 \ge& \sum_{e\in T\backslash S} \frac{2n|T\backslash S|}{\lambda_{n}(\LL(S))\lambda_{n}(\LL({S\cup \{e\}}))} \cdot \frac{\lambda_2(\LL(S))\lambda_2(\LL(S\cup T))}{2n|T\backslash S|}\\
		\ge& \kh{\frac{ \lambda_2(\LL)}{\lambda_{n}(\LL(Q))}}^2 =  \kh{\frac{ \lambda_2(\LL)}{n}}^2.
	\end{align*}
 The last equality holds since the largest eigenvalue of the Laplacian matrix of a complete graph with $n$ nodes is $n$.


	Similarly, we derive the upper bound of the curvature $\alpha$. Let $j$ be any candidate edge in $S\backslash T$. Then, we have
	\begin{align*}
		&\frac{ g_{\{j\}}(S\cup T \backslash \{j\})}{ g_{\{j\}}(S\backslash \{j\})} \\
  \ge& \frac{2n|T\backslash S|}{\lambda_{n}(\LL({S\cup T \backslash \{j\}}))\lambda_{n}(\LL(S\cup T))}\cdot \frac{\lambda_{1}(\LL({S\backslash \{j\}}))\lambda_{1}(\LL(S))}{2n|T\backslash S|}\\
  \ge& \kh{\frac{ \lambda_2(\LL)}{\lambda_{n}(\LL({Q}))}}^2 = \kh{\frac{ \lambda_2(\LL)}{n}}^2,
	\end{align*}
	which combining with the curvature definition completes the proof of Equation~\eqref{eq:upalpha}.
\end{proof}

Pursuing the enhancement of the theoretical efficiency of this greedy algorithm, a larger approximation ratio $(1-e^{-\gamma\alpha})/\alpha$ is favorable. This entails that a superior $\gamma$ and an inferior $\alpha$ could augment the magnitude of $(1-e^{-\gamma\alpha})/\alpha$. In addition, we discern from Theorem~\ref{lem:sub} that the boundaries of $\gamma$ and $\alpha$ are intrinsically related to $\lambda_2(\LL)$. Predominantly, denser graphs are associated with a larger $\lambda_2(\LL)$~\cite{GhBo06}, suggesting an enlarged lower bound for $\gamma$ and a diminished upper bound for $\alpha$. Thus, denser graphs potentially allow for a tighter approximation ratio bound.


\subsection{Enhanced Greedy Algorithm via Two Approximation Algorithms}

Algorithm~\ref{alg:det} is primarily impeded by the time-consuming step of inverting the Laplacian matrix. While this step guarantees an accurate computation of the marginal decrease of the Kirchhoff index at every stage, it hampers the scalability of the greedy algorithm. Inspired by~\cite{LiZh18}, we introduce two approximation algorithms to expedite the computation of the marginal decrease.

In every iteration, the marginal decrease of the Kirchhoff index $\Delta(e) = K(\calG) -K(\{e\})$ can be rewritten as $\frac{\bb_e^\top \LL^{2\dag} \bb_e}{1+\bb_e^\top \LL^{\dag} \bb_e}$ for any $e \in Q$~\cite{SuShLyDo15,PrKoMe22}.
The key to computing the marginal decrease is the calculation of the effective resistance, $\bb_e^\top \LL^{\dag} \bb_e$, and the square of the biharmonic distance, $\bb_e^\top \LL^{2\dag} \bb_e$, between all non-incident pairs of nodes in $\calG$. The biharmonic distance between two nodes $u$ and $v$, denoted as $d(u,v)$, is defined as $d(u,v)^2 = (\ee_u - \ee_v)^\top \LL^{2\dag} (\ee_u - \ee_v)$~\cite{LiRuFu10}. For simplicity, we refer to $d(u,v)^2$ as the biharmonic distance.

To circumvent matrix inversion, the authors of~\cite{SpSr08} proposed an algorithm, \textsc{ERComp}($\calG$,$Q$,$\epsilon$), which constructs a data structure in $\tilde{O}(m\epsilon^{-2})$ time, where $\tilde{O}(\cdot)$ suppresses ${\rm poly}(\log n)$ factors, and then returns the approximate effective resistance for each node pair incident to the edge $e\in Q$ in $O(\log n\epsilon^{-2})$ time with relative error $\epsilon$. Thus, we need $\tilde{O}(m\epsilon^{-2} + |Q|\epsilon^{-2})$ time to query the effective resistance between each node pair connected by the candidate edge $e \in Q$. Later, the authors of~\cite{YiShLiZh18} utilized this idea to propose another algorithm, \textsc{BDComp}($\calG$,$Q$,$\epsilon$), for evaluating the biharmonic distance. Similarly, the total time cost for querying biharmonic distance between node pairs in set $Q$ is $\tilde{O}(m\epsilon^{-2}+|Q|\epsilon^{-2})$.

Utilizing these two efficient algorithms, we can obtain two sets $\{(e,\tilde{r}(e))|e \in Q\} = \textsc{ERComp}(\calG,Q,\epsilon/2)$ and $\{(e,\tilde{b}(e))|e \in Q\} = \textsc{BDComp}(\calG,Q,\epsilon/2)$, which respectively represent the approximated effective resistance and the biharmonic distance for each node pair incident to the candidate edge $e\in Q$. Consequently, $\tilde{r}(e) \approx_{\epsilon/2} \bb_e^\top \LL^\dag \bb_e$ and $\tilde{b}(e) \approx_{\epsilon/2} \bb_e^\top \LL^{2\dag} \bb_e$ hold for any candidate edge $e\in Q$. Based on these approximations, we can infer that $\tilde{\Delta}(e) = \frac{\tilde{b}(e)}{1+\tilde{r}(e)}$ constitutes an $\epsilon$-approximation of $\Delta(e)$ for any $e \in Q$, that is $\tilde{\Delta}(e) \approx_{\epsilon}\Delta(e)$.

Building on this analysis, we propose an approximation algorithm as described in Algorithm~\ref{alg:JLsol}. The quality of this algorithm is discussed in Theorem~\ref{thm:JLsol}.

\begin{theorem}\label{thm:JLsol}
For any $k>0$ and error $\epsilon$, Algorithm~\ref{alg:JLsol} operates in $\tilde{O}(km\epsilon^{-2}+k|Q|\epsilon^{-2})$ time. It produces a solution $T$ by iteratively selecting $k$ edges, where for the edge chosen in each round, the marginal decrease of this edge is an $\epsilon$-approximation of the maximal marginal decrease of the Kirchhoff index.
\end{theorem}

While the idea that combines Laplacian solvers with Johnson-Lindenstrauss random projection has been exploited in~\cite{BlWaNaWa23,PrKoMe22}, we provide a more detailed analysis of the approximation error. We demonstrate that, with appropriate settings, an accurate approximation of the marginal decrease in the Kirchhoff index can be achieved with a relative error guarantee, offering a stronger theoretical result compared to the conclusions in~\cite{BlWaNaWa23,PrKoMe22}.


Despite the efficiency of the enhanced greedy algorithm, a persistent problem arises because we need to query $|Q|$ different marginal decreases. This process can be notably time-consuming, especially for sparse graphs. Although alternative methods have been proposed to approximate effective resistance~\cite{PeLoYo21,LiYu23,YaTa23,LiLiDa23}, the huge time cost resulted from quadratically many queries remains a challenge and is also a main limitation of the state-of-the-art algorithms~\cite{PrKoMe22}. Consequently, in the subsequent sections, we will explore ways to reduce the number of queries in an attempt to expedite the optimization algorithm.

\begin{small}
\begin{algorithm}[tb]
\begin{small}
	\caption{\textsc{Approx}$(\calG, S_1, S_0, Q, k, \epsilon)$}
    \label{alg:JLsol}
	\Input{
		A connected graph $\calG=(V,E)$;
        a candidate edge set $Q=(V\times V)\backslash E$;
        an integer $1 \leq k \leq |Q|$;
        a real number $\epsilon>0$\\
	}
	\Output{
		$T$: A subset of $T$ with $|T|=k$
	}
Initialize solution $T = \emptyset$ \;
\For{$i = 1$ to $k$}{
 $\{(e,\tilde{r}(e))|e \in Q\} \leftarrow \textsc{ERComp}(\calG,Q,\epsilon/2)$\;
 $\{(e,\tilde{b}(e))|e \in Q\} \leftarrow \textsc{BDComp}(\calG,Q,\epsilon/2)$\;
Select $e_i$ s.t. $e_i \gets \mathrm{arg\, max}_{e \in Q} \tilde{\Delta}(e) =  \frac{\tilde{B}_T(e)}{1+\tilde{R}_T(e)}$\;
Update solution $T \gets T \cup \{ e_i \}$ \;
Update the graph $\calG \gets \calG(V, E \cup \{ e_i \})$ \;
Update the candidate edge set $Q \gets Q\backslash \{e_i\}$
}
\Return $T$.
\end{small}
\end{algorithm}
\end{small}

\section{Nearly Linear Time Algorithm}
In the application of the heuristic greedy algorithm, pinpointing the maximum marginal decrease, $\Delta(e)$, for each $e\in Q$ in every iteration poses a significant challenge, especially when considering that the candidate edge set $Q$ can potentially scale up to $\Omega(n^2)$. This potential size underscores the pressing need for efficient pruning strategies. In the ensuing discussions, we introduce a suite of gradient-based algorithms. Concurrently, we detail methods designed to prune the large original candidate set down to a smaller subset $P$, where $P \subset Q$ and $|P| \ll |Q|$, thereby facilitating fewer query operations. By leveraging these methodologies, we are able to devise efficient algorithms that yield approximate solutions to the central problem.

\subsection{Transforming Decrement Search to Gradient Maximization}
Traditionally, heuristic greedy algorithms select edges based on the maximal marginal decrease in the Kirchhoff index.
While the marginal decrease is commonly used, the gradient or partial derivative of a function has also been explored for its ability to measure edge modification impact, as seen in~\cite{YiShLiZh18,SiBoBaMo18,LiZh18}. However, its application in optimization remains limited, with notable exceptions like~\cite{ZhZhCh21}. Several studies have examined the gradient of the Kirchhoff index~\cite{LiZh18,GhBoSa08,YaMaQiWe18,YiShLiZh18}, which captures the significance of an edge through its potential to alter the Kirchhoff index following modification of the edge within the graph. Inspired by these insights, we adopt the gradient of the Kirchhoff index to characterize the significance of an edge instead of the marginal decrease $\Delta(e)$, and we will provide a new approach for solving this problem.

First, we show how to express the gradient of the Kirchhoff index.

Initially, we extend the simple unweighted graph $\calG=(V,E)$ to the weighted graph $\calH=(V,V\times V,w)$, wherein $w$ is the edge weight function which assigns a weight of 1 to $e\in E$ and 0 otherwise. Thus, the addition of a new edge $e\in (V\times V)\backslash E$ to the graph $\calG$ can be interpreted as altering its edge weight $w(e)$ from 0 to 1.

Subsequently, the Kirchhoff index of graph $\calH$ can be expressed as a function of the weight $w$ of each edge, and the gradient or partial derivative $c(e)$ for each edge $e\in (V\times V)\backslash E$ can be determined by~\cite{GhBoSa08,YiShLiZh18},
\begin{equation*}
c(e)=\frac{\partial K(\calH)}{\partial w(e)} = n \bb_e^\top \LL^{2\dag} \bb_e.
\end{equation*}

Since the gradient of the Kirchhoff index of each edge also measures the importance of each edge, we can naturally propose a gradient-based greedy algorithm by iteratively finding one edge with the largest gradient and adding it to the graph. The algorithm is outlined in Algorithm~\ref{alg:grad} with its time complexity being $O(n^3+kn^2)$, by utilizing a similar analysis as algorithm \textsc{Deter}. Although this gradient-based algorithm \textsc{Grad} does not possess a theoretical quality guarantee, subsequent empirical evidence illustrates comparable efficacy between \textsc{Grad} and \textsc{Deter}.

\begin{small}
\begin{algorithm}[tb]
\begin{small}
	\caption{\textsc{Grad}$(\calG, S_1, S_0, Q, k)$}
    \label{alg:grad}
	\Input{
		A connected graph $\calG=(V,E)$;
        a candidate edge set $Q=(V\times V)\backslash E$;
        an integer $1 \leq k \leq |Q|$\\
	}
	\Output{
		$T$: A subset of $T$ with $|T|=k$
	}
Compute $\LL^{\dag}$\;
Compute $\LL^{2\dag}$\;
Initialize solution $T = \emptyset$ \;
\For{$i = 1$ to $k$}{
Compute $\Delta(e)= K(\calG)-K(\{e\})$ for each $e \in Q$\;
Select $e_i$ s.t. $e_i \gets \mathrm{arg\, max}_{e \in Q} c(e)$\;
Update solution $T \gets T \cup \{ e_i \}$ \;
Update the graph $\calG \gets \calG(V, E \cup \{ e_i \})$ \;
Update $\LL^{\dag} \gets \LL^{\dag} - \frac{ \LL^\dag\bb_e \bb_e^\top\LL^\dag }{1+\bb_e^\top \LL^\dag \bb_e}$\;
Update $\LL^{2\dag} \gets \LL^{2\dag} + \frac{\LL^\dag \bb_e \bb_e^\top\LL^{2\dag}\bb_e\bb_e^\top \LL^\dag}{(1+\bb_e^\top\LL^\dag \bb_e)^2} - \frac{\LL^{2\dag}\bb_e\bb_e^\top \LL^\dag}{1+\bb_e^\top\LL^\dag \bb_e} - \frac{\LL^{\dag}\bb_e\bb_e^\top \LL^{2\dag}}{1+\bb_e^\top\LL^\dag \bb_e}$\;
Update the candidate edge set $Q \gets Q \backslash \{e_i\} $ \;}
\Return $T$.
\end{small}
\end{algorithm}
\end{small}

Although we employ the gradient instead of the marginal decrease of the Kirchhoff index for designing the algorithm, the inherent problem of high time complexity incurred by querying $|Q|$ gradients still exists. Rather than computing the gradient $c(e)$ for all candidate edges $e\in Q$—a process that can be computationally expensive—we will propose a pruning technique using some properties based on geometry in the following sections.




\subsection{Primary Nearly Linear-time Solution}

\subsubsection{Employing Convex Hull Approximation}
For the sake of simplicity and subsequent discussions, we redefine $c(e)=\bb_e^\top \LL^{2\dag}\bb_e$ for any edge $e=(i,j)\in Q$, since we do not change the number of nodes in the graph. Then $c(e) = \bb_e^\top \LL^{2\dag}\bb_e = \smallnorm{\LL^\dag (\ee_i-\ee_j)}_2^2$ can be regarded as the squared Euclidean distance between nodes $i$ and $j$~\cite{GhBoSa08,YiShLiZh18}. Consequently, we introduce a point set $P=\{p_1,p_2,\ldots,p_n\}$ with each point's coordinate specified as $p_i= \LL^\dag \ee_i \in \mathbb{R}^n$. Hence, the task of maximizing the gradient of $c(e)$ is equivalent to identifying the most distant pair of points within the set $P$.

A na\"{i}ve approach to finding the furthest point pair involves computing all pairwise distances, which is computationally costly.
In examining the spatial distribution of a given set of points, we observed that certain points are, relatively speaking, located within the interior of this point set, while others lie on the boundary. It is a general observation that the distances among interior points are shorter than the distances between those on the boundary. Recognizing this, our primary focus is to identify these boundary points, as doing so would facilitate faster computations in subsequent analyses.
This leads us to leverage the concept of the convex hull to discern the boundary points.

\begin{definition}~\cite{Ro70}
Given a set of $n$ points $P=\{p_1,p_2,\ldots,p_n\}$, its convex hull, denoted as $C(P)$, is the intersection of all convex supersets of $P$. Each $p \in P$ for which $p \notin C(P \backslash \{p\})$ is called an extremal point of $C(P)$. The collection of all extremal points constitutes the extremal point set $X(P) \subseteq P$.
\end{definition}

Remarkably, there exists a unique bijective relationship between the convex hull and the extremal point set of any given set of points. More precisely, given the convex hull of a set of points, one can unambiguously determine its corresponding extremal point set, and vice versa.
Such an equivalence has profound implications, since algorithms that compute the convex hull can be immediately leveraged to identify the extremal point set, and conversely, algorithms that determine the extremal point set can be adapted to infer the convex hull.
Based on this equivalence, we introduce a critical property of the convex hull,
as formulated in the following lemma.
\begin{lemma}~\cite{Ro70}
Given a set of $n$ points $P=\{p_1,p_2,\ldots,p_n\}$, let $C(P)$ represent the convex hull of $P$ and $X(P)$ denote the extremal point set of $C(P)$, we have $d(P)=d(C(P))=d(X(P))$, where $d(\cdot)$ signifies the diameter of a set.
\end{lemma}

In general, the extremal point set is smaller than the original point set. Leveraging the properties of the extremal point set of the convex hull, we can identify the furthest point pair by computing distances between points within the extremal point set instead of between all point pairs. This dramatically reduces computational time, particularly if the extremal points can be efficiently determined. Below, we elaborate on how to efficiently identify the extremal point set or the convex hull of a given point set.

For a set of points $P=\{p_1,p_2,\ldots,p_n\} \subset \mathbb{R}^d$, the time complexity to find the convex hull of these points is $O(n^{\lfloor d/2\rfloor})$\cite{Ch93}.  
Hence, computing the convex hull in high-dimensional space is costly. Our analysis above yields $n$ points $p_i=\LL^\dag \ee_i \in \mathbb{R}^n$ for $i=1,2,\ldots,n$, signifying that finding the convex hull of these $n$ points incurs a time complexity of $O(n^{\lfloor n/2\rfloor})$. This is impractical due to the excessive time cost. However, an algorithm called \textsc{ApproxConv} can approximate the convex hull with significantly lower time complexity~\cite{AwKaZh18,AwKaZh20,Ka15}. The quality of the algorithm \textsc{ApproxConv} is described in the following lemma:

\begin{lemma}\label{lem:appconv}~\cite{AwKaZh18,AwKaZh20,Ka15}
For a set of points $P=\{p_1,p_2,\ldots,p_n\} \subset \mathbb{R}^d$, and a parameter $\mu\in (0,1)$, the algorithm $\textsc{ApproxConv}(P,\mu)$ generates a subset $\bar{X}(P) \subseteq X(P)$ of $l=|\bar{X}(P)|$ points, where $X(P)$ is the set of extreme points of $S$. The algorithm has a time complexity of $O(nl(d+\mu^{-2}))$ and ensures that the Euclidean distance for any $p\in X(P)$ to the convex hull of the point set $\bar{X}(P)$ does not exceed $\mu d(P)$, where $d(P)$ represents the diameter of the point set $P$.
\end{lemma}

Based on Lemma~\ref{lem:appconv}, we can conveniently obtain a point set $\bar{X}(P) \subseteq X(P)$ as an approximate extreme point set by applying the algorithm $\textsc{ApproxConv}(P,\mu)$ for any $\mu \in (0,1)$. We thus have:
\begin{equation*}
d(P) = d(X(P)) = d(C(P)) \leq 2\mu d(P) + d(\bar{X}(P)),
\end{equation*}
implying that
\begin{equation}\label{equ:cvapp}
d(P)^2 \approx_{8\mu} d(\bar{X}(P))^2.
\end{equation}
This approximation affords us a significant reduction in computation time, since finding the furthest point pairs from set $\bar{X}(P)$ now only needs $O(l^2n)$ time complexity.


\subsubsection{Coordinate Projection}

Based on Lemma~\ref{lem:appconv} and Equation~\eqref{equ:cvapp}, the \textsc{ApproxConv} algorithm can effectively generate an approximation of the extremal point set for point set $P$. Nonetheless, its application still requires a time complexity of $\Omega(n^2)$, since the dimension of each node is $n$.
Given that the exclusion of points from set $P$ is not feasible, optimizing the \textsc{ApproxConv} algorithm necessitates a reduction in the dimensionality of each point. In pursuit of enhanced computational efficiency, we leverage the Johnson-Lindenstrauss (JL) lemma~\cite{JoLi84,Ac01}. This pivotal lemma provides a means to lower the dimensionality without significantly changing the pairwise distances between the points. Therefore, we can reduce the processing time of \textsc{ApproxConv} greatly.


First, we introduce the Johnson-Lindenstrauss Lemma~\cite{Ac01}.

\begin{lemma}
\label{lemma:JL}
	Given fixed vectors $\vvv_1,\vvv_2,\ldots,\vvv_n\in \mathbb{R}^d$ and
	$\beta>0$, let
 $\QQ_{t\times d}$, $t\ge 24\log n/\beta^2$, be a matrix, with each entry being either $1/\sqrt{t}$ or $-1/\sqrt{t}$ with the same probability $1/2$.
	 Then, with a probability of at least $1-1/n$,
	\[(1-\beta)\|\vvv_i-\vvv_j\|^2_2\le \|\QQ \vvv_i-\QQ \vvv_j\|^2_2\le
	(1+\beta)\|\vvv_i-\vvv_j\|^2_2\] for all pairs $i,j\le n$.
\end{lemma}

Lemma~\ref{lemma:JL} allows the coordinates of a point $p_i$ to be projected from an $n$-dimensional space to a lower-dimensional one, approximately preserving the distances between each pair of points. Let us denote a new point set $\tilde{P}=\{\tilde{p}_1,\tilde{p}_2,\ldots,\tilde{p}_n\}$, where each point $\tilde{p}_i$ possesses a new coordinate $\tilde{p}_i=\QQ\LL^{\dag}\ee_i \in \mathbb{R}^t$ for $i=1,2,\ldots,n$.
Then, for any $i,j$, we have $\norm{p_i-p_j}_2^2 \approx_{\beta} \norm{\tilde{p}_i-\tilde{p}_j}_2^2$, which leads to
\begin{equation}\label{equ:d_JL}
d(P)^2 \approx_{\beta} d(\tilde{P})^2.
\end{equation}

Subsequently, we can apply the \textsc{ApproxConv} algorithm on the point set $\tilde{P}$ with an error $\mu$, which in turn generates a subset $\tilde{X}(\tilde{P})$ of the point set $\tilde{P}$ in $O(n\tilde{l}(\log n/\beta^2+\mu^2))$ time complexity, where $\tilde{l}=|\tilde{X}(\tilde{P})|$. According to Equation~\eqref{equ:cvapp}, this subset meets the condition
\begin{equation}\label{equ:d_JLcov}
d(\tilde{P})^2 \approx_{8\mu} d(\tilde{X}(\tilde{P}))^2.
\end{equation}
By integrating Equations~\eqref{equ:d_JL} and~\eqref{equ:d_JLcov}, we have
\begin{equation}
d(\tilde{X}(\tilde{P}))^2 \approx_{\beta+8\mu} d(P)^2,
\end{equation}
implying that the diameter of the point set $P$ can be approximated by projecting the coordinates of points and approximating the convex hull of the resulting points.
After determining the extremal point set $\tilde{X}(\tilde{P})$, we can now find the furthest point pair from set $\tilde{X}(\tilde{P})$ in $O(\tilde{l}^2\log n/\beta^2)$ time.

\subsubsection{Coordinate Computation}
Given a set of points, the above analysis shows us how to approximate its diameter efficiently. However, there still exists a crucial consideration of computation of the projected coordinate. The calculation of $\QQ\LL^{\dag}\ee_i$ requires matrix inversion, an operation that typically demands $O(n^3)$ time for directly computing the pseudoinverse of the Laplacian matrix $\LL$. To circumvent this expensive operation, we utilize the fast SDD linear system solver~\cite{SpTe14,CoKyMiPaJaPeRaXu14}, as detailed below.

\begin{lemma}\label{lem:solver}~\cite{SpTe14,CoKyMiPaJaPeRaXu14}
\label{solve-lemma}
Let $\TT \in \mathbb{R}^{n \times n}$ be a semidefinite positive symmetric matrix with nonzero $m$ entries, $\bb \in \mathbb{R}^n$ a vector, and $\gamma > 0$ an error parameter. There exists a solver, denoted by $\aaa = \SDDMSolver(\TT, \bb, \gamma)$, which yields a vector $\aaa \in \mathbb{R}^n$ satisfying $\norm{\aaa - \TT^{-1} \bb}_{\TT} \leq \gamma \norm{\TT^{-1} \bb}_{\TT}$. The expected runtime of this solver is $\tilde{O}(m)$, where $\tilde{O}(\cdot)$ suppresses ${\rm poly}(\log n)$ factors.
\end{lemma}

Building on this solver, we can approximate the coordinates of each point $\tilde{p}_i$. Let $\XX=\QQ\LL^\dag$, and $\hat{\XX}_{j,:}=\SDDMSolver(\LL,\QQ_{j,:},\gamma)$. Then, $\hat{p}_i=\hat{\XX}\ee_i$ becomes an approximation of the projected coordinate of point $\tilde{p}_i$ in an $O(\log n)$-dimensional space, calculated by the Laplacian solver. We now aim to show that the distances between the approximate coordinates retain an accurate approximation of the distances between the original coordinates. Lemma~\ref{lem:solvecomp} verifies that this approximation can be achieved efficiently using \textsc{Solve} (see Lemma~\ref{solve-lemma}).

\begin{lemma}\label{lem:solvecomp}
Assume that
\begin{equation}\label{equ:lem1}
\norm{\XX(\ee_u-\ee_v)} \approx_{\beta} \norm{\XX(\ee_u-\ee_v)}
\end{equation}
holds for every pair $u,v\in V$, and
\begin{equation}\label{equ:lem2}
\smallnorm{\hat{\XX}_{j,:}-\XX_{j,:}}_{\LL} \leq \gamma \smallnorm{\XX_{ j,:}}_{\LL},
\end{equation}
where $\gamma \leq \frac{\delta}{3n} \sqrt{\frac{6(1-\beta)}{n(n^2-1)(1+\beta)}}$, holds for any $j=1,2,\ldots,t$. Then,
\begin{equation}\label{equ:lem3}
\smallnorm{\XX(\ee_u-\ee_v)}^2_2 \approx_{\delta} \smallnorm{\hat{\XX}(\ee_u-\ee_v)}_2^2
\end{equation}
holds for every pair $u,v\in V$.
\end{lemma}

We omit the proof since it is similar to that in~\cite{YiShLiZh18}.

Lemma~\ref{lem:solvecomp} allows us to approximate the projected coordinates of each point in $\tilde{O}(m/\beta^2)$ and then constructs a set of points $\hat{P}=\{\hat{p}_1,\hat{p}_2,\ldots,\hat{p}_n\}$ with approximate projected coordinates being $\hat{p}_i=\XX\ee_i \in \mathbb{R}^t$. Consequently, we obtain
\begin{equation}
d(\tilde{P})^2 \approx_{\delta} d(\hat{P})^2.
\end{equation}

By applying the algorithm \textsc{ApproxConv} with the approximated projected point set $\hat{P}$ computed in $\tilde{O}(m/\beta^2)$, the algorithm returns a subset $\hat{X}(\hat{P})$ of point set $\hat{P}$ in $O(n\hat{l}(\log n/\beta^2+\mu^2))$, where $\hat{l}=|\hat{X}(\hat{P})|$, satisfying
\begin{align}
 d(\hat{X}(\hat{P}))^2  \approx_{8\mu+\beta+\delta} d(P)^2.
\end{align}
Since the furthest point pair of set $\hat{X}(\hat{P})$ can be computed in $O(\hat{l}^2 \log n /\beta^2)$, we now can efficiently determine the furthest point pairs.

\subsubsection{Approximation Algorithm}

Equipped with the convex hull approximation and coordinate projection, we now delineate an efficient algorithm for Problem~\ref{pro:kim}, as detailed in Algorithm~\ref{alg:grad}.

The structure of Algorithm~\ref{alg:grad} consists of $k$ iterative rounds (Lines 5-13), each focusing on the selection of an individual edge. The steps within each round are systematically organized as follows: (1) Random matrix generation: A random matrix $\QQ_{t\times n}$ is constructed in $O(n\log n\epsilon^{-2})$ time (Line 6). (2) Coordinate approximation: Utilizing the JL lemma and the Laplacian solver, the approximated projected coordinates of points $\hat{p}_1,\hat{p}_2,\ldots,\hat{p}_n$ are computed in $\tilde{O}(m \epsilon^{-2})$ time (Lines 7-9). (3) Extremal point determination: The extremal point set of the point set $\hat{P}$ is computed in $O(n\log n \epsilon^{-2} \hat{l})$ time (Line 10). (4) Maximal distance finding: The maximal distance between the points within the extremal point set is identified in $O(\hat{l}^2 \log n\epsilon^{-2})$ time (Line 11). (5) Solution and graph update: Finally, the solution and graph are updated based on the obtained results (Lines 12-13).

Hence, the total time complexity of Algorithm~\ref{alg:fastgrad} is $\tilde{O}(km \hat{l}/\epsilon^2)$. The following theorem rigorously characterizes the quality of Algorithm~\ref{alg:fastgrad}.

\begin{theorem}\label{Performance}
Given any positive integer $k$ and an error parameter $\epsilon \in (0,1)$, the run-time of Algorithm~\ref{alg:fastgrad} is $\tilde{O}(km \hat{l}/\epsilon^2)$. Through the iterative selection of $k$ edges, the algorithm produces a solution $T$. Furthermore, the gradient of the edge selected in each round constitutes an $\epsilon$-approximation of the maximal gradient of the objective function.
\end{theorem}

\begin{small}
\begin{algorithm}[tb]
\begin{small}
	\caption{\textsc{FastGrad}$(\calG, \epsilon,k)$}
    \label{alg:fastgrad}
	\Input{
		A connected graph $\calG=(V,E)$; a real number $\epsilon>0$; an integer $k$\\
	}
	\Output{
		$T$: A subset of $(V \times V)\backslash E$ satisfying $|T|=k$
	}
 Initialize solution $T = \emptyset$\;
 Set $\mu = \epsilon/24$, $\beta = \epsilon/3$, $\delta = \epsilon/3$\;
 Set $\gamma =  \frac{\delta}{3n} \sqrt{\frac{6(1-\beta)}{n(n^2-1)(1+\beta)}}$ \;
 $t = \lceil \log(n)/\beta^2 \rceil$\;
 $\LL \leftarrow $ the Laplacian matrix of graph $\calG$\;
 \For{$i=1$ to $k$}{
 Generate random Gaussian matrices $\QQ_{t\times n}$\;
 \For{$j=1$ to $t$}{
    $\XX_{j,:}$ = \textsc{Solve} ($\LL,\QQ_{j,:},\gamma$)
    }
    $\hat{P} = \{\hat{p}_1,\hat{p}_2,\ldots,\hat{p}_n\}$ with the coordinate of each point being $\XX_{:,1},\XX_{:,2},\ldots,\XX_{:,n}$\;
 $H = \textsc{ApproxConv}(\hat{P},\mu)$\;
 $(x,y) \leftarrow \argmax_{u,v\in H} \norm{\XX_{:,u}-\XX_{:,v}}_2^2$\;
 Update solution $T \leftarrow T \cup \{(x,y)\}$\;
 Update the graph $\calG \leftarrow \calG(V,E\cup \{(x,y)\})$\;
 }

\Return $T$.
\end{small}
\end{algorithm}
\end{small}

\subsection{Faster Algorithm by Pre-pruning and Update}

While Algorithm~\ref{alg:fastgrad} shows significant efficiency, we identify certain redundancies. The most time-consuming steps in each iteration are step (2), computing the approximated projected coordinates, and step (3), invoking the \textsc{ApproxConv} routine. In what follows, we introduce improvements to accelerate the algorithm by addressing these two aspects.

\subsubsection{Fast Update of the Coordinate}

Repeatedly calling the Laplacian solver $t=\lceil\log n/\epsilon^{2}\rceil$ times to compute the approximated projected coordinates in each iteration is a significant computational overhead. We recognize that after selecting an edge in each iteration, we can efficiently update the coordinate of each node by employing the Sherman-Morrison formula~\cite{Me73} rather than recomputing them.

For each node $i\in V$, its projected coordinate is $\QQ \LL^\dag \ee_i$. After adding edge $e=(x,y)$ to the graph, the coordinate updates to $\QQ (\LL+\bb_e\bb_e^\top)^\dag \ee_i = \QQ\LL^\dag \ee_i-\frac{\QQ\LL^\dag\bb_e \bb_e^\top \LL^\dag\ee_i}{\bb_e^\top \LL^\dag \bb_e}$ according to Sherman-Morrison formula. Instead of re-approximating the projected coordinates, we only need to compute $\xx=\bb_e^\top \LL^\dag$ once, then we can update the coordinate of each point to be $\QQ\LL^\dag \ee_i-\frac{\QQ\xx \xx^\top\ee_i}{\bb_e^\top \xx}$ for each $i\in V$ by matrix manipulation. Direct computation of vector $\xx=\bb_e^\top \LL^\dag$ relies on matrix inversion and is time-consuming. So we use the Laplacian solver a single time in $\tilde{O}(m)$ to get an approximation of the vector $\xx$ by $\tilde{\xx}=\textsc{Solve}(\LL,\bb_e,\delta_0)$, where $\delta_0$ can be sufficiently small for accuracy. Consequently, the coordinates of the nodes can be then updated efficiently in $O(n\log n/\epsilon^{-2})$ time, rather than using Laplacian solver for $\lceil\log n\epsilon^{-2}\rceil$ times.

\begin{remark}
    Although the use of the Laplacian solver may lead to some inaccuracy, it has very accurate quality in practice since $\delta_0$ can be very small. Also, we only use the Laplacian solver once in each iteration. So we believe that this update process does not lead to more errors.
\end{remark}





\subsubsection{Preprocessing Techniques for Extremal Point Selection}

The time complexity of approximating the extremal point set of the convex hull is correlated with the number of nodes in the network. As the extremal point set consists of a small-sized subset of the total $n$ points, our aim is to prune as many non-extremal points as possible before calling the \textsc{ApproxConv} routine, thereby saving time. It seems infeasible to directly remove some points from the point set. However, by combining the information of the graph structure, we will prune as many internal points as possible.

An essential property of extremal points is that the point furthest from any given point must itself be an extremal point. The distance considered here refers to the Euclidean distance in space, correlating with the biharmonic distance on graphs. Determining whether a node is or is not the furthest node against any other node on the graph is time-consuming. To expedite this, we employ node centrality with two prerequisites: (1) it encapsulates the graph's furthest distance information, (2) it can be computed efficiently.

The eccentricity centrality, related to the shortest path distance, serves our purpose.
We use the algorithm \textsc{EccComp}($\calG$) in~\cite{LiQiQiChZhLi22} to calculate the eccentricities of all nodes.
If any node's eccentricity is smaller than any one of its neighbor's eccentricity, this node is a central node and we prune this node. 
Since this step prunes those inner nodes, it does not influence the approximate accuracy of our algorithm which computes extremal points. Moreover, after adding a few edges to the graph, the central nodes should not be affected substantially and remain central, so we can simply preprocess this step once to prune those central nodes before using algorithm \textsc{ApproxConv}.

Some prior research suggests that node pairs with maximal effective resistance decrease the Kirchhoff index of a graph significantly~\cite{VaDeCe17,WaPoKoVa14}. Other studies sample nodes based on their diagonal entries in $\LL^\dag$, an effective resistance-based metric that captures the total effective resistances between the node and all other nodes~\cite{PrKoMe22}. However, the focus on the largest distance in the convex hull does not match this approach well. Moreover, the computational expense of calculating $\LL^\dag$ diagonal values outweighs that of eccentricities~\cite{AnPrVaMe20,LiQiQiChZhLi22}. Hence, to avoid excessive time on preprocessing, we select eccentricities as the key metric to prune central nodes.

\subsubsection{Faster Algorithm}

Armed with techniques of preprocessing of the convex hull approximation and coordinate update, we are now in position to propose a more efficient algorithm for Problem~\ref{pro:kim}, which is depicted in Algorithm~\ref{alg:fastgrad+}. While it is not an exact algorithm, it still can effectively  solve Problem~\ref{pro:kim}.

The structure of Algorithm~\ref{alg:fastgrad+} consists of a preprocessing step (Lines 1-13) and $k$ iterative rounds (Lines 14-23).
It first computes the eccentricities of all nodes (Line 5) in $O(m)$ time and then selects nodes that are not the central nodes (Lines 6-9) in $O(m)$ time. Then it computes the approximate projected coordinates of each node in $\tilde{O}(m\epsilon^{-2})$ time. Next, the steps within each round are systematically organized as follows: (1) Extremal point determination: The extremal point set of the point set $\hat{P}$ is determined in $O(n\log n \epsilon^{-2} \hat{l})$ time (Line 15-16). (2) Maximal distance finding: The maximal distance between the points within the extremal point set is identified in $O(\hat{l}^2 \log n\epsilon^{-2})$ time (Line 17). (3) Coordinates, solution and graph update: Finally, the coordinates are updated (Line 18-20) in $\tilde{O}(m + n\epsilon^{-2})$ and solution and graph are updated based on the obtained results (Lines 18-23).

Hence, the total time complexity of Algorithm~\ref{alg:fastgrad+} is  $\tilde{O}(km \hat{l}/\epsilon^2)$. The following theorem rigorously characterizes the quality of Algorithm~\ref{alg:fastgrad+}.

\begin{theorem}\label{thm:fastgrad+}
Given any positive integer $k$ and an error parameter $\epsilon \in (0,1)$, Algorithm~\ref{alg:fastgrad+} executes in $\tilde{O}(km \hat{l}/\epsilon^2)$ time. Through the iterative selection of $k$ edges, the algorithm produces a solution $T$. Furthermore, the gradient of the edge selected in each round constitutes an $\epsilon$-approximation of the maximal gradient of the objective function.
\end{theorem}

Since our pruning and updating process does not affect the accuracy  of the theoretical analysis of our algorithm, while reducing the running time in practice, Algorithm \textsc{FastGrad+} is more efficient than Algorithm \textsc{FastGrad}, and both achieve similar quality.

\begin{small}
\begin{algorithm}[tb]
\begin{small}
	\caption{\textsc{FastGrad+}$(\calG, \epsilon,k)$}
    \label{alg:fastgrad+}
	\Input{
		A connected graph $\calG=(V,E)$; a real number $\epsilon>0$; an integer $k$\\
	}
	\Output{
		$T$: A subset of $(V \times V)\backslash E$ satisfying $|T|=k$
	}
 Initialize solution $T = \emptyset$\;
 Set $\mu = \epsilon/24$, $\beta = \epsilon/3$, $\delta = \epsilon/3$\;
 Set $\gamma =  \frac{\delta}{3n} \sqrt{\frac{6(1-\beta)}{n(n^2-1)(1+\beta)}}$ \;
 $t = \lceil \log(n)/\beta^2 \rceil$\;
 $\{(i,p_i)\} \leftarrow \textsc{EccComp}(\calG)$\;
 $Y =\emptyset$\;
 \For{$i\in V$}{
    \If{$p_i \geq p_j, j \in N_i$}{
        $Y \leftarrow Y \cup \{i\}$
    }
 }
 $\LL \leftarrow $ the Laplacian matrix of graph $\calG$\;
   Generate random Gaussian matrices $\QQ_{t\times n}$\;
 \For{$j=1$ to $t$}{
    $\XX_{i,:}$ = \textsc{Solve} ($\LL,\QQ_{i,:}, \gamma$)
    }
 \For{$i=1$ to $k$}{
 $\hat{P} = \{\hat{p}_i | i\in Y \}$ with the coordinate of each point being $\XX_{:,i}$\;
 $H = \textsc{ApproxConv}(\hat{P},\mu)$\;
 $(x,y) \leftarrow \argmax_{u,v\in H} \norm{\XX_{:,u}-\XX_{:,v}}_2^2$\;
 $\yy $ = \textsc{Solve} ($\LL,\ee_x-\ee_y, \gamma$) \;
 \For{$i \in Y$}{
 $\XX_{i,:} \leftarrow \XX_{i,:} - \frac{\QQ\yy \yy^\top \ee_i}{\yy^\top (\ee_x-\ee_y)}$
 }
 Update solution $T \leftarrow T \cup \{(x,y)\}$\;
 Update the graph $\calG \leftarrow \calG(V,E\cup \{(x,y)\})$\;
 Update the Laplacian matrix $\LL$\;
 }
\Return $T$.
\end{small}
\end{algorithm}
\end{small}

\subsection{Convex Hull Approximation for Only Once}

Although we have reduced the time complexity of our algorithm by fast updates of the coordinate and pruning the inner nodes, we still need to use the algorithm \textsc{ApproxConv} in each iteration, which takes much time. Since we add a small number of edges to the graph, the coordinate of each node does not change significantly during the edge adding process, which means that the extremal point set also does not change substantially after edge addition. So, we can compute the convex hull only once to get the extremal point set, then find edges from this reduced set for $k$ iterations.
This is the most efficient algorithm proposed in the present paper, and the quality is illustrated in the following theorem.

\begin{theorem}\label{thm:alg3}
Given any positive integer $k$ and an error parameter $\epsilon \in (0,1)$, Algorithm~\ref{alg:oneconv} executes in $\tilde{O}(km \hat{l}\epsilon^{-2})$ time. Through the iterative selection of $k$ edges, the algorithm produces a solution $T$.
\end{theorem}

\begin{small}
\begin{algorithm}[tb]
\begin{small}
	\caption{\textsc{OneConv}$(\calG, \epsilon,k)$}
    \label{alg:oneconv}
	\Input{
		A connected graph $\calG=(V,E)$; a real number $\epsilon>0$; an integer $k$\\
	}
	\Output{
		$T$: A subset of $(V \times V)\backslash E$ satisfying $|T|=k$
	}
 Initialize solution $T = \emptyset$\;
 Set $\mu = \epsilon/24$, $\beta = \epsilon/3$, $\delta = \epsilon/3$\;
 Set $\gamma =  \frac{\delta}{3n} \sqrt{\frac{6(1-\beta)}{n(n^2-1)(1+\beta)}}$ \;
 $t = \lceil \log(n)/\beta^2 \rceil$\;
 $\LL \leftarrow $ the Laplacian matrix of graph $\calG$\;
 Generate random Gaussian matrices $\QQ_{t\times n}$\;
 \For{$j=1$ to $t$}{
    $\XX_{i,:}$ = \textsc{Solve} ($\LL,\QQ_{i,:}, \gamma$)
    }
     $\hat{P} = \{\hat{p}_i | i\in Y \}$ with the coordinate of each point being $\XX_{:,i}$\;
     $H = \textsc{ApproxConv}(\hat{P},\mu)$\;
 \For{$i=1$ to $k$}{
 $(x,y) \leftarrow \argmax_{u,v\in H} \norm{\XX_{:,u}-\XX_{:,v}}_2^2$\;
 $\yy $ = \textsc{Solve} ($\LL,\ee_x-\ee_y, \gamma$) \;
 \For{$i \in H$}{
 $\XX_{i,:} \leftarrow \XX_{i,:} - \frac{\QQ\yy \yy^\top \ee_i}{\yy^\top (\ee_x-\ee_y)}$
 }
 Update solution $T \leftarrow T \cup \{(x,y)\}$\;
 Update the graph $\calG \leftarrow \calG(V,E\cup \{(x,y)\})$\;
 Update the Laplacian matrix $\LL$\;
 }
\Return $T$.
\end{small}
\end{algorithm}
\end{small}





\section{Experiments}

\subsection{Experiment Setup}

\textbf{Data Description.} We utilize ten datasets of undirected networks from the Network Repository~\cite{RoAh15} and SNAP~\cite{LeSo16}. Our experiments focus on the largest connected component (LCC) of these networks. Table~\ref{tab:1} provides a detailed statistical overview of the LCC for each network. Notably, the biggest LCC encompasses over 5 million nodes and 12 million edges.

\begin{table}[h]
	\centering
		\caption{Networks and the numbers of nodes and edges of their LCC.}\label{tab:1}
\fontsize{8}{8}\selectfont			
\begin{tabular}{m{1.2cm}<{\raggedright}m{0.8cm}<{\raggedleft}m{0.8cm}<{\raggedleft}|m{1.2cm}<{\raggedright}m{1.0cm}<{\raggedleft}m{1.0cm}<{\raggedleft}}
				\toprule
				Networks & Nodes & Edges & Networks & Nodes & Edges \\
				\midrule
                EmailUniv &1133&5451 & EmailEU &32430 &54397\\
                BitcoinAlpha  &3783&24186& Douban &154908 &327162 \\
                Gnutella08 &6301 &20777& Delicious &536108 &1365961\\
                Government &7057 &89429&  YoutubeSnap& 1134890 &2987624\\
                Anybeat &12645& 67053& DBLP &5624219 &12282055\\
				\bottomrule
			\end{tabular}
\end{table}






\textbf{Experimental Setup.} All experiments were executed on a Linux server equipped with a 4.2 GHz Intel i9-9900K CPU and 128G memory. We restricted our experiments to a single-threaded environment. The experiments were implemented using the \textit{Julia} programming language.

\textbf{State-of-the-Art Algorithm Implementation.} The work of~\cite{PrKoMe22} introduced several algorithms tailored for the Kirchhoff index optimization problem. Among these, the \textsc{ColStochJLT} algorithm stands out as the state-of-the-art, demonstrating superior runtime and quality metrics in comparison to other baseline algorithms. We adopted the settings and configurations as delineated in~\cite{PrKoMe22} for our experiments. Subsequent sections will present the quality results of the \textsc{ColStochJLT} algorithm.

\textbf{Edge Addition Strategy Comparison.} We compare the quality of various algorithms: \textsc{Deter}, \textsc{Grad}, \textsc{FastGrad}, \textsc{FastGrad+}, \textsc{OneConv}, and \textsc{ColStochJLT}. It is worth noting that we have excluded the \textsc{Approx} algorithm from this comparison. Although it is an enhancement of \textsc{Deter}, it does not address the computational challenge at hand.

\begin{figure}[t]
		\centering
		\includegraphics[width=0.9\linewidth]{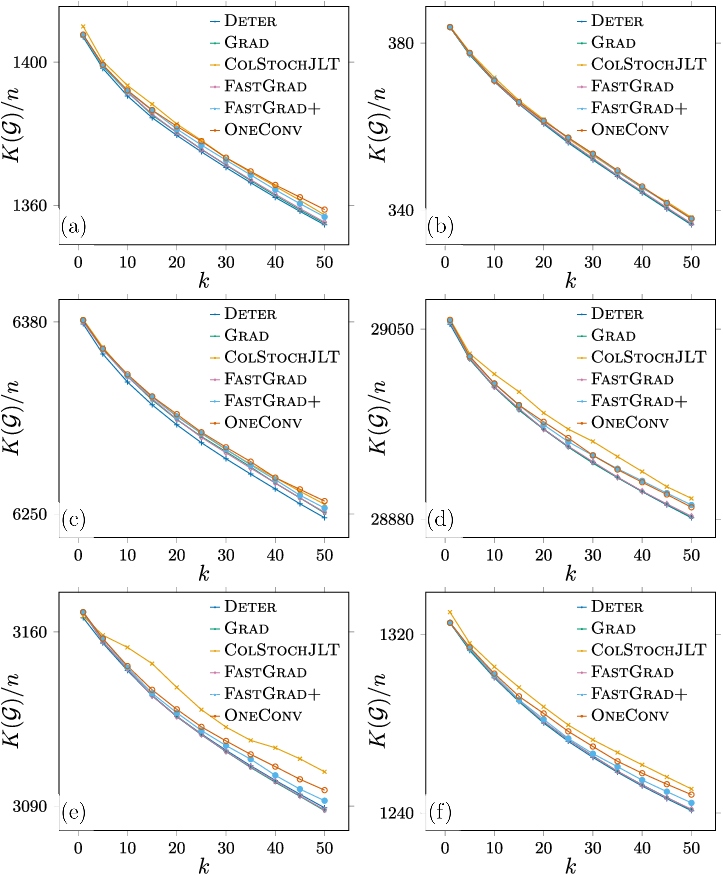}
		\caption{Kirchhoff index divided by $n$ returned by all algorithms for $k=1,2,\ldots,50$ on six small networks, BitcoinAlpha (a), EmailUniv (b), Anybeat (c), EmailEU (d), Gnutella08 (e), and Government (f). \label{fig:1}}
\end{figure}

\begin{figure}[t]
		\centering
		\includegraphics[width=0.9\linewidth]{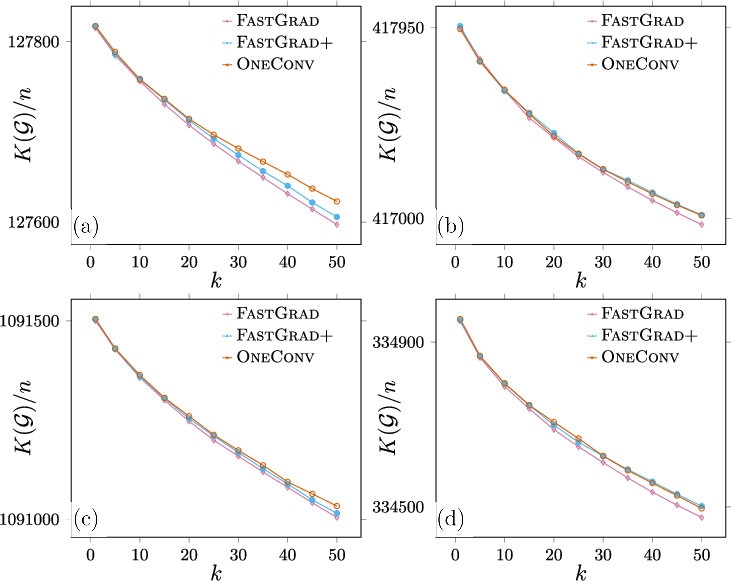}
		\caption{Kirchhoff index divided by $n$ returned by algorithms \textsc{FastGrad}, \textsc{FastGrad+}, and \textsc{OneConv}, for $k=1,2,\ldots,50$ on four large networks, Douban (a), Delicious (b), YoutubeSnap (c), and DBLP (d). \label{fig:2}}
\end{figure}

\begin{figure}[t]
		\centering
		\includegraphics[width=0.9\linewidth]{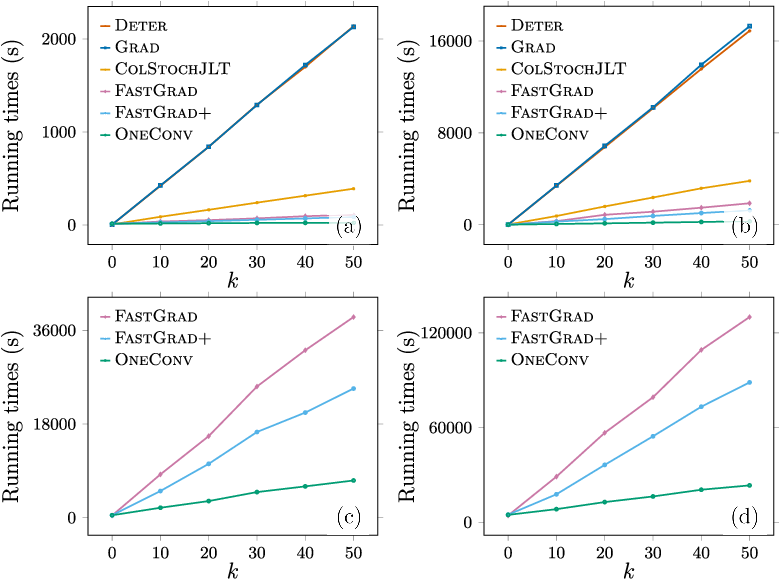}
		\caption{Running time of different algorithms for $k=10,20,\ldots,50$ on four networks, Anybeat (a), EmailEU (b), YoutubeSnap (c), and DBLP (d). \label{fig:3}}
\end{figure}

\subsection{Algorithms Comparison}

We investigate the impact of the number of added edges, denoted as $k$, on the quality of our proposed algorithms. We vary $k$ in the range $5, 10, \ldots, 50$, setting parameters $\mu=0.01$, $\beta=0.1$, and $\gamma=0.1$ without directly specifying $\epsilon$. For smaller networks (those with fewer than 100,000 nodes), we execute all algorithms and present their quality in Figure~\ref{fig:1}. Notably, the gradient-based algorithm \textsc{Grad} exhibits quality comparable to the standard greedy algorithm \textsc{Deter}. This observation suggests that identifying a single edge with the highest gradient is an effective strategy for the Kirchhoff index optimization problem. Furthermore, our newly introduced algorithms, namely \textsc{FastGrad}, \textsc{FastGrad+}, and \textsc{OneConv}, demonstrate comparable quality. For larger networks, those with more than 100,000 nodes, \textsc{ColStochJLT} is infeasible due to its prohibitive computational time, as noted in~\cite{PrKoMe22}. In contrast, our algorithms remain operational. As depicted in Figure~\ref{fig:2}, all three of these algorithms exhibit good quality.

Figure~\ref{fig:3} delineates the runtime of various algorithms across four networks for $k$ values ranging from 10 to 50 in increments of 10. Among the algorithms, our proposed methods consistently outperform in terms of computational efficiency, with \textsc{OneConv} emerging as the most time-efficient. Specifically, on the largest network, which comprises over 5 million nodes and 12 million edges, \textsc{FastGrad} requires approximately 36 hours, \textsc{FastGrad+} takes around 24 hours, while \textsc{OneConv} completes in  mere 3 hours. It is worth noting that both \textsc{FastGrad+} and \textsc{OneConv} demand more preprocessing time compared to \textsc{FastGrad}, yet they compensate with reduced overall runtime.

Our empirical findings align well with our theoretical analyses. Given that \textsc{FastGrad} possesses the highest time complexity, it is unsurprising that it demands the longest runtime. Conversely, \textsc{OneConv},  with its lowest time complexity among all our proposed algorithms, ensures the quickest execution. However, while we offer error guarantees for \textsc{FastGrad}, \textsc{OneConv} does not provide such guarantees. This trade-off elucidates why \textsc{FastGrad} delivers superior quality, whereas \textsc{OneConv} has the shortest runtime.

\section{Related Work}

Efficient computation of effective resistances is essential in a large range of applications. This has initiated the introduction of many algorithms for estimating the resistance distances between node pairs. A method anchored in random projection was introduced in~\cite{SpSr08,MaGaKo15}, targeting the resistance distances between endpoints of each edge, which was subsequently optimized in~\cite{HaAkYo16} through the integration of Wilson’s algorithm~\cite{Wi96}. The work presented in~\cite{PeLoYo21} uses a localized algorithm to deduce pairwise effective resistance, leveraging random walk samplings and the formulation of random spanning trees. Moreover, Monte Carlo techniques were employed in~\cite{LiLiDa23, YaTa23} to augment the efficacy of preceding algorithms, as delineated in~\cite{PeLoYo21}. While the present landscape offers a suite of methods for computing resistance distances between individual nodes, the task of computing the Kirchhoff index remains challenging, necessitating calculations involving resistance distances for quadratically many node pairs.


A significant body of research has delved into the problem of adding new edges to graphs to minimize the Kirchhoff index~\cite{KlRa93, GhBoSa08}. The authors in~\cite{KoAc23} tackled the NP-hard problem of selecting a maximum of $k$ edges for this purpose. Interestingly, their proof establishes the NP-hardness of the Kirchhoff index optimization problem via a reduction from the well-known 3-colorability problem. Several prior studies, including~\cite{SuShLyDo15, YaMaQiWe18, HuHeTa19, PrKoMe22, BlWaNaWa23}, have investigated the optimization of the Kirchhoff index by the addition of a predetermined number of edges, from an algorithmic perspective. Initial greedy algorithms proposed in this domain exhibited a time complexity of $O(kn^3)$~\cite{SuShLyDo15, YaMaQiWe18}. However, advancements by~\cite{PrKoMe22} have since whittled this down to a more efficient $\tilde{O}(n^2+km)$. In a separate line of work,~\cite{GhBoSa08} delved into the problem of determining edge weights in an existing weighted network to curtail the total effective resistance. Drawing from this, the authors of~\cite{WaPoKoVa14} formulated both upper and lower bounds for total effective resistance, considering both edge addition and deletion operations. The problem of optimally introducing a single edge to graphs, defined by a set number of nodes and diameter, to minimize total effective resistance was carefully investigated in~\cite{ElSpVaJaKo11}. However, a notable drawback of these methodologies remains their limited scalability, especially to large-scale graphs with strong expansion properties.



Besides the Kirchhoff index, several novel modifications to it have emerged over recent years. Notably, these include the multiplicative degree-Kirchhoff index~\cite{ChZh07} and the additive degree-Kirchhoff index~\cite{GuFeYu12}.
A notable revelation from the literature is the equivalence of the multiplicative degree-Kirchhoff index $K^*(\mathcal{G})$ of a graph $\mathcal{G}$ to four times the edge count in graph $\calG$ multiplied by the graph's Kemeny constant~\cite{ChZh07}. The Kemeny constant itself has diverse applications across various domains~\cite{Hu14}. For instance, it is harnessed as a yardstick for gauging user navigation efficiency within the vast expanse of the World Wide Web~\cite{LeLo02}. Further applications extend to measuring surveillance efficiency in robotic networks~\cite{PaAgBu15} and characterizing the noise robustness intrinsic to specific formation control protocols~\cite{JaOl19}. Due to its wide applications, the Kemeny constant has attracted much recent attention from the data mining and data management communities~\cite{XuShZhKaZh20,XiZh24SIGMOD, XiZh24KDD}.



In practical graph editing, the operation of edge addition has been ubiquitously employed, catering to a myriad of application purposes.
Noteworthy instances include enhancing the centrality of a particular node~\cite{CrDaSeVe16,DaOlSe19,ShYiZh18}, and amplifying the quantity of spanning trees~\cite{LiPaYiZh20}. Recent forays into graph theory have witnessed the edge addition operation's utility in reducing the network diameter~\cite{FrGaGuMa15,AdGi22}. Parallel endeavors have explored graph augmentation through edge additions, targeting diverse goals such as bolstering algebraic connectivity~\cite{GhBo06}. However, it is imperative to note that these existing methodologies and approaches aren't directly translatable to the distinct challenge of minimizing the Kirchhoff index through edge addition.



\section{Conclusion}
In this paper, we studied the Kirchhoff index minimizing problem by adding $k$ new edges. We first proposed greedy algorithms to solve this problem, providing relative error guarantees with respect to the optimal solution, based on the bounds of the submodularity ratio and the curvature. Then, we introduced a gradient-based greedy algorithm as a new paradigm to solve this problem. We leveraged geometric properties, convex hull approximation, and approximation of the projected coordinate of each point, to provide an efficient algorithm that approximates the gradient with a relative error guarantee and has nearly-linear time complexity.  We used pre-pruning and fast update techniques to further reduce the time complexity. We also proposed an algorithm by using the convex hull approximation only once. Our proposed algorithms have linear time complexity. Extensive experimental results on ten real-life networks demonstrated that our proposed algorithms outperform the state-of-the-art methods in terms of efficiency and effectiveness. Moreover, our fast algorithms are scalable to large graphs with over 5 million nodes and 12 million edges.

There are several potential avenues for future work. It would be interesting to leverage other novel approximation and sparsification techniques to speed up our proposed algorithms further while preserving the accuracy guarantee. Furthermore, we believe that the proof techniques developed in the present work can be exploited to provide more efficient algorithms for other relevant network-based optimization problems.
Note that the Kirchhoff index depends on the sum of non-zero eigenvalues of the Laplacian matrix $\LL$. Some other quantities are also dependent on serval or all eigenvalues of $\LL$, such as condition number of $\LL$~\cite{ClMoStWi79,SpTe11}, the number of spanning trees of a graph~\cite{LiPaYiZh20}. It is expected that our algorithms can be extended to various optimization problems involving these quantities through edge addition.
In addition, our techniques could also prove useful for tackling the problems of optimizing Kemeny's constant and resistance distance.

\bibliographystyle{IEEEtran}
\normalem
\bibliography{tkde}

\providecommand{\noopsort}[1]{}\providecommand{\singleletter}[1]{#1}
\begin{thebibliography}{10}
\providecommand{\url}[1]{#1}
\csname url@samestyle\endcsname
\providecommand{\newblock}{\relax}
\providecommand{\bibinfo}[2]{#2}
\providecommand{\BIBentrySTDinterwordspacing}{\spaceskip=0pt\relax}
\providecommand{\BIBentryALTinterwordstretchfactor}{4}
\providecommand{\BIBentryALTinterwordspacing}{\spaceskip=\fontdimen2\font plus
\BIBentryALTinterwordstretchfactor\fontdimen3\font minus
  \fontdimen4\font\relax}
\providecommand{\BIBforeignlanguage}[2]{{%
\expandafter\ifx\csname l@#1\endcsname\relax
\typeout{** WARNING: IEEEtran.bst: No hyphenation pattern has been}%
\typeout{** loaded for the language `#1'. Using the pattern for}%
\typeout{** the default language instead.}%
\else
\language=\csname l@#1\endcsname
\fi
#2}}
\providecommand{\BIBdecl}{\relax}
\BIBdecl

\bibitem{SpSr08}
D.~A. Spielman and N.~Srivastava, ``Graph sparsification by effective
  resistances,'' in \emph{Proc. 40th Annu. {ACM} Symp. Theory Comput.}, 2008,
  pp. 563--568.

\bibitem{ChKeMaSpTe11}
P.~Christiano, J.~A. Kelner, A.~Madry, D.~A. Spielman, and S.-H. Teng,
  ``Electrical flows, {L}aplacian systems, and faster approximation of maximum
  flow in undirected graphs,'' in \emph{Proc. 43rd Annu. {ACM} Symp. Theory
  Comput.}, 2011, pp. 273--282.

\bibitem{MaStTa15}
A.~Madry, D.~Straszak, and J.~Tarnawski, ``Fast generation of random spanning
  trees and the effective resistance metric,'' in \emph{Proc. 26th Annu.
  {ACM-SIAM} Symp. Discrete Algorithms}, 2015, pp. 2019--2036.

\bibitem{AlAnLaOv18}
V.~L. Alev, N.~Anari, L.~C. Lau, and S.~Oveis~Gharan, ``Graph clustering using
  effective resistance,'' in \emph{9th Innovations in Theoretical Computer
  Science Conference}, 2018, pp. 41:1--41:16.

\bibitem{QiHa07}
H.~Qiu and E.~R. Hancock, ``Clustering and embedding using commute times,''
  \emph{IEEE Trans. Pattern Anal. Mach. Intell.}, vol.~29, no.~11, pp.
  1873--1890, 2007.

\bibitem{FoPiReSa07}
F.~Fouss, A.~Pirotte, J.-m. Renders, and M.~Saerens, ``Random-walk computation
  of similarities between nodes of a graph with application to collaborative
  recommendation,'' \emph{IEEE Trans. Knowl. Data. Eng.}, vol.~19, no.~3, pp.
  355--369, 2007.

\bibitem{CaZhCh18}
H.~Cai, V.~W. Zheng, and K.~C.-C. Chang, ``A comprehensive survey of graph
  embedding: Problems, techniques, and applications,'' \emph{IEEE Trans. Knowl.
  Data. Eng.}, vol.~30, no.~9, pp. 1616--1637, 2018.

\bibitem{BePaPr08}
R.~Behmo, N.~Paragios, and V.~Prinet, ``Graph commute times for image
  representation,'' in \emph{2008 IEEE Conf. Comput. Vis. Pattern Recog.}\hskip
  1em plus 0.5em minus 0.4em\relax IEEE, 2008, pp. 1--8.

\bibitem{ShYiZh18}
L.~Shan, Y.~Yi, and Z.~Zhang, ``Improving information centrality of a node in
  complex networks by adding edges,'' in \emph{27th Int. Joint Conf. Artif.
  Intell.}, 2018, pp. 3535--3541.

\bibitem{LiPeShYiZh19}
H.~Li, R.~Peng, L.~Shan, Y.~Yi, and Z.~Zhang, ``Current flow group closeness
  centrality for complex networks,'' in \emph{Proc. World Wide Web Conf.},
  2019, pp. 961--971.

\bibitem{QiDhTaPeWa21}
J.~Qiu, L.~Dhulipala, J.~Tang, R.~Peng, and C.~Wang, ``Lightne: A lightweight
  graph processing system for network embedding,'' in \emph{Proc. 2021 Int.
  Conf. Manage. Data}, 2021, pp. 2281--2289.

\bibitem{ShMaWuCh14}
J.~Shi, N.~Mamoulis, D.~Wu, and D.~W. Cheung, ``Density-based place clustering
  in geo-social networks,'' in \emph{Proc. 2014 Int. Conf. Manage. Data}, 2014,
  pp. 99--110.

\bibitem{DoSibu18}
F.~D{\"o}rfler, J.~W. Simpson-Porco, and F.~Bullo, ``Electrical networks and
  algebraic graph theory: {M}odels, properties, and applications,'' \emph{Proc.
  IEEE}, vol. 106, no.~5, pp. 977--1005, 2018.

\bibitem{KlRa93}
D.~J. Klein and M.~Randi{\'c}, ``Resistance distance,'' \emph{J. Math. Chem.},
  vol.~12, no.~1, pp. 81--95, 1993.

\bibitem{GhBoSa08}
A.~Ghosh, S.~Boyd, and A.~Saberi, ``Minimizing effective resistance of a
  graph,'' \emph{SIAM Rev.}, vol.~50, no.~1, pp. 37--66, 2008.

\bibitem{ShZh19}
Y.~Sheng and Z.~Zhang, ``Low-mean hitting time for random walks on
  heterogeneous networks,'' \emph{IEEE Trans. Inf. Theory}, vol.~65, no.~11,
  pp. 6898--6910, 2019.

\bibitem{TiLe10}
A.~Tizghadam and A.~Leon-Garcia, ``Autonomic traffic engineering for network
  robustness,'' \emph{IEEE J. Sel. Areas Commun.}, vol.~28, no.~1, 2010.

\bibitem{WoLiCh16}
F.~M.~F. Wong, Z.~Liu, and M.~Chiang, ``On the efficiency of social recommender
  networks,'' \emph{IEEE/ACM Trans. Netw.}, vol.~24, no.~4, pp. 2512--2524,
  2016.

\bibitem{PaBa14}
S.~Patterson and B.~Bamieh, ``Consensus and coherence in fractal networks,''
  \emph{IEEE Trans. Control Netw. Syst.}, vol.~1, no.~4, pp. 338--348, 2014.

\bibitem{QiZhYiLi19}
Y.~Qi, Z.~Zhang, Y.~Yi, and H.~Li, ``Consensus in self-similar hierarchical
  graphs and {S}ierpi{\'n}ski graphs: {C}onvergence speed, delay robustness,
  and coherence,'' \emph{IEEE Trans. Cybern.}, vol.~49, no.~2, pp. 592--603,
  2019.

\bibitem{YiZhPa20}
Y.~Yi, Z.~Zhang, and S.~Patterson, ``Scale-free loopy structure is resistant to
  noise in consensus dynamics in complex networks,'' \emph{IEEE Trans.
  Cybern.}, vol.~50, no.~1, pp. 190--200, 2020.

\bibitem{LiZhZeZh23}
C.~Liu, X.~Zhou, A.~N. Zehmakan, and Z.~Zhang, ``A fast algorithm for
  moderating critical nodes via edge removal,'' \emph{IEEE Trans. Knowl. Data
  Eng.}, vol.~3, no.~4, pp. 1385--1398, 2024.

\bibitem{BlWaNaWa23}
M.~Black, Z.~Wan, A.~Nayyeri, and Y.~Wang, ``Understanding oversquashing in
  {GNN}s through the lens of effective resistance,'' in \emph{Proc. Int. Conf.
  Mach. Learn.}, 2023, pp. 2528--2547.

\bibitem{LiZh18}
H.~Li and Z.~Zhang, ``{K}irchhoff index as a measure of edge centrality in
  weighted networks: {N}early linear time algorithms,'' in \emph{Proc. 29th
  Annu. {ACM-SIAM} Symp. Discrete Algorithms}, 2018, pp. 2377--2396.

\bibitem{YiShLiZh18}
Y.~Yi, L.~Shan, H.~Li, and Z.~Zhang, ``Biharmonic distance related centrality
  for edges in weighted networks.'' in \emph{Proc. 27th Int. Joint Conf. Artif.
  Intell.}, 2018, pp. 3620--3626.

\bibitem{SaHaAbSh21}
A.~Said, S.-U. Hassan, W.~Abbas, and M.~Shabbir, ``{Netki: A Kirchhoff index
  based statistical graph embedding in nearly linear time},''
  \emph{Neurocomputing}, vol. 433, pp. 108--118, 2021.

\bibitem{CuXiZh19}
P.~Cui, X.~Wang, J.~Pei, and W.~Zhu, ``A survey on network embedding,''
  \emph{IEEE Trans. Knowl. Data Eng.}, vol.~31, no.~5, pp. 833--852, 2019.

\bibitem{FrYaKuToCh23}
S.~Freitas, D.~Yang, S.~Kumar, H.~Tong, and D.~H. Chau, ``{Graph vulnerability
  and robustness: A survey},'' \emph{IEEE Trans. Knowl. Data Eng.}, vol.~35,
  no.~6, pp. 5915--5934, 2022.

\bibitem{SuShLyDo15}
T.~Summers, I.~Shames, J.~Lygeros, and F.~D{\"o}rfler, ``Topology design for
  optimal network coherence,'' in \emph{Proc. 2015 Eur. Control Conf.}, 2015,
  pp. 575--580.

\bibitem{PrKoMe22}
M.~Predari, R.~Kooij, and H.~Meyerhenke, ``Faster greedy optimization of
  resistance-based graph robustness,'' in \emph{Proc. 2022 IEEE/ACM Int. Conf.
  Adv. Soc. Netw. Anal. Min.}, 2022, pp. 1--8.

\bibitem{YaMaQiWe18}
C.~Yang, J.~Mao, X.~Qian, and P.~Wei, ``Designing robust air transportation
  networks via minimizing total effective resistance,'' \emph{IEEE Trans.
  Intell. Transp. Syst.}, vol.~20, no.~6, pp. 2353--2366, 2018.

\bibitem{HuHeTa19}
G.~Huang, W.~He, and Y.~Tan, ``Theoretical and computational methods to
  minimize kirchhoff index of graphs with a given edge k-partiteness,''
  \emph{Appl. Math. Comput.}, vol. 341, pp. 348--357, 2019.

\bibitem{AbDa11}
A.~Das and D.~Kempe, ``Submodular meets spectral: Greedy algorithms for subset
  selection, sparse approximation and dictionary selection,'' in \emph{Proc.
  28th Int. Conf. Mach. Learn.}, 2011, pp. 1057--1064.

\bibitem{BiBuKrTs17}
A.~A. Bian, J.~M. Buhmann, A.~Krause, and S.~Tschiatschek, ``Guarantees for
  greedy maximization of non-submodular functions with applications,'' in
  \emph{Proc. 34th Int. Conf. Mach. Learn.}, 2017, pp. 498--507.

\bibitem{ElSpVaJaKo11}
W.~Ellens, F.~M. Spieksma, P.~Van~Mieghem, A.~Jamakovic, and R.~E. Kooij,
  ``Effective graph resistance,'' \emph{Linear Algebra Appl.}, vol. 435,
  no.~10, pp. 2491--2506, 2011.

\bibitem{KoAc23}
R.~E. Kooij and M.~A. Achterberg, ``Minimizing the effective graph resistance
  by adding links is np-hard,'' \emph{arXiv preprint arXiv:2302.12628}, 2023.

\bibitem{PiSo22}
C.~Pizzuti and A.~Socievole, ``Incremental computation of effective graph
  resistance for improving robustness of complex networks: A comparative
  study,'' in \emph{Int. Conf. Complex Netw. Their Appl.}, 2022, pp. 419--431.

\bibitem{Me73}
C.~D. Meyer, Jr, ``Generalized inversion of modified matrices,'' \emph{SIAM J.
  Appl. Math.}, vol.~24, no.~3, pp. 315--323, 1973.

\bibitem{SuShLyDo17correct}
\BIBentryALTinterwordspacing
T.~Summers, I.~Shames, J.~Lygeros, and F.~Dorfler, ``Correction to “topology
  design for optimal network coherence”,'' 2017. [Online]. Available:
  \url{https://personal.utdallas.edu/~ths150130/papers/ECC_Correction.pdf}
\BIBentrySTDinterwordspacing

\bibitem{GhBo06}
A.~Ghosh and S.~P. Boyd, ``Growing well-connected graphs,'' in \emph{45th
  {IEEE} Conference on Decision and Control}.\hskip 1em plus 0.5em minus
  0.4em\relax {IEEE}, 2006, pp. 6605--6611.

\bibitem{LiRuFu10}
Y.~Lipman, R.~M. Rustamov, and T.~A. Funkhouser, ``Biharmonic distance,''
  \emph{ACM Trans. Graph.}, vol.~29, no.~3, p.~27, 2010.

\bibitem{PeLoYo21}
P.~Peng, D.~Lopatta, Y.~Yoshida, and G.~Goranci, ``Local algorithms for
  estimating effective resistance,'' in \emph{Proc. 27th ACM SIGKDD Conf.
  Knowl. Discov. Data Min.}, 2021, pp. 1329--1338.

\bibitem{LiYu23}
Z.~Liu and W.~Yu, ``Computing effective resistances on large graphs based on
  approximate inverse of cholesky factor,'' in \emph{Proc. 2023 Des. Autom.
  Test Eur. Conf. Exhib.}\hskip 1em plus 0.5em minus 0.4em\relax IEEE, 2023,
  pp. 1--6.

\bibitem{YaTa23}
R.~Yang and J.~Tang, ``Efficient estimation of pairwise effective resistance,''
  \emph{Proc. ACM Manage. Data}, vol.~1, no.~1, pp. 1--27, 2023.

\bibitem{LiLiDa23}
M.~Liao, R.-H. Li, Q.~Dai, H.~Chen, H.~Qin, and G.~Wang, ``Efficient resistance
  distance computation: {T}he power of landmark-based approaches,'' \emph{Proc.
  ACM Manage. Data}, vol.~1, no.~1, pp. 1--27, 2023.

\bibitem{SiBoBaMo18}
M.~Siami, S.~Bolouki, B.~Bamieh, and N.~Motee, ``Centrality measures in linear
  consensus networks with structured network uncertainties,'' \emph{IEEE Trans.
  Control Netw. Syst.}, vol.~5, no.~3, pp. 924--934, 2017.

\bibitem{ZhZhCh21}
Z.~Zhang, Z.~Zhang, and G.~Chen, ``Minimizing spectral radius of
  non-backtracking matrix by edge removal,'' in \emph{Proc. 30th ACM Int. Conf.
  Inf. Knowl. Manag.}, 2021, pp. 2657--2667.

\bibitem{Ro70}
R.~T. Rockafellar, \emph{Convex Analysis}.\hskip 1em plus 0.5em minus
  0.4em\relax Princeton university press, 1970, vol.~18.

\bibitem{Ch93}
B.~Chazelle, ``An optimal convex hull algorithm in any fixed dimension,''
  \emph{Discrete Comput. Geom.}, vol.~10, no.~4, pp. 377--409, 1993.

\bibitem{AwKaZh18}
P.~Awasthi, B.~Kalantari, and Y.~Zhang, ``Robust vertex enumeration for convex
  hulls in high dimensions,'' in \emph{Int. Conf. Artif. Intell. Stat.}, 2018,
  pp. 1387--1396.

\bibitem{AwKaZh20}
------, ``Robust vertex enumeration for convex hulls in high dimensions.''
  \emph{Ann. Oper. Res.}, vol. 295, no.~1, pp. 37--74, 2020.

\bibitem{Ka15}
B.~Kalantari, ``A characterization theorem and an algorithm for a convex hull
  problem,'' \emph{Ann. Oper. Res.}, vol. 226, pp. 301--349, 2015.

\bibitem{JoLi84}
W.~B. Johnson and J.~Lindenstrauss, ``{Extensions of Lipschitz mappings into a
  Hilbert space},'' \emph{Contemp. Math.}, vol.~26, pp. 189--206, 1984.

\bibitem{Ac01}
D.~Achlioptas, ``Database-friendly random projections,'' in \emph{Proc. 20th
  ACM SIGMOD-SIGACT-SIGART Symp. Princ. Database Syst.}, 2001, pp. 274--281.

\bibitem{SpTe14}
D.~A. Spielman and S.-H. Teng, ``Nearly linear time algorithms for
  preconditioning and solving symmetric, diagonally dominant linear systems,''
  \emph{SIAM J. Matrix Anal. Appl.}, vol.~35, no.~3, pp. 835--885, 2014.

\bibitem{CoKyMiPaJaPeRaXu14}
M.~B. Cohen, R.~Kyng, G.~L. Miller, J.~W. Pachocki, R.~Peng, A.~B. Rao, and
  S.~C. Xu, ``Solving {SDD} linear systems in nearly $m \log^{1/2} n$ time,''
  in \emph{Proc. 46th Annu. ACM Symp. Theory Comput.}\hskip 1em plus 0.5em
  minus 0.4em\relax ACM, 2014, pp. 343--352.

\bibitem{LiQiQiChZhLi22}
W.~Li, M.~Qiao, L.~Qin, L.~Chang, Y.~Zhang, and X.~Lin, ``On scalable
  computation of graph eccentricities,'' in \emph{Proc. 2022 Int. Conf. Manage.
  Data}, 2022, pp. 904--916.

\bibitem{VaDeCe17}
P.~Van~Mieghem, K.~Devriendt, and H.~Cetinay, ``{Pseudoinverse of the Laplacian
  and best spreader node in a network},'' \emph{Phys. Rev. E}, vol.~96, no.~3,
  p. 032311, 2017.

\bibitem{WaPoKoVa14}
X.~Wang, E.~Pournaras, R.~E. Kooij, and P.~Van~Mieghem, ``Improving robustness
  of complex networks via the effective graph resistance,'' \emph{Eur. Phys. J.
  B}, vol.~87, pp. 1--12, 2014.

\bibitem{AnPrVaMe20}
E.~Angriman, M.~Predari, A.~van~der Grinten, and H.~Meyerhenke, ``Approximation
  of the diagonal of a {L}aplacian’s pseudoinverse for complex network
  analysis,'' in \emph{Proc. 28th Annu. Eur. Symp. Algorithms}, vol. 173, 2020,
  pp. 1--24.

\bibitem{RoAh15}
R.~Rossi and N.~Ahmed, ``The network data repository with interactive graph
  analytics and visualization,'' in \emph{Proc. 29th AAAI Conf. Artif.
  Intell.}, 2015, pp. 4292--4293.

\bibitem{LeSo16}
J.~Leskovec and R.~Sosi{\v{c}}, ``{SNAP}: A general-purpose network analysis
  and graph-mining library,'' \emph{ACM Trans. Intel. Syst. Tec.}, vol.~8,
  no.~1, p.~1, 2016.

\bibitem{MaGaKo15}
C.~Mavroforakis, R.~Garcia-Lebron, I.~Koutis, and E.~Terzi, ``Spanning edge
  centrality: {L}arge-scale computation and applications,'' in
  \emph{Proceedings of the 24th International Conference on World Wide Web},
  2015, pp. 732--742.

\bibitem{HaAkYo16}
T.~Hayashi, T.~Akiba, and Y.~Yoshida, ``Efficient algorithms for spanning tree
  centrality,'' in \emph{Proc. Int. Joint Conf. Artif. Intell.}, 2016, pp.
  3733--3739.

\bibitem{Wi96}
D.~B. Wilson, ``Generating random spanning trees more quickly than the cover
  time,'' in \emph{Proceedings of the Twenty-Eighth Annual ACM Symposium on
  Theory of Computing}, 1996, pp. 296--303.

\bibitem{ChZh07}
H.~Chen and F.~Zhang, ``Resistance distance and the normalized laplacian
  spectrum,'' \emph{Discrete Appl. Math.}, vol. 155, no.~5, pp. 654--661, Mar.
  2007.

\bibitem{GuFeYu12}
I.~Gutman, L.~Feng, and G.~Yu, ``Degree resistance distance of unicyclic
  graphs,'' \emph{Trans. Comb.}, vol.~1, no.~2, pp. 27--40, Jun. 2012.

\bibitem{Hu14}
J.~J. Hunter, ``The role of {K}emeny's constant in properties of {M}arkov
  chains,'' \emph{Commun. Stat.-Theory Methods}, vol.~43, no.~7, pp.
  1309--1321, Apr. 2014.

\bibitem{LeLo02}
M.~Levene and G.~Loizou, ``{K}emeny's constant and the random surfer,''
  \emph{Amer. Math. Monthly}, vol. 109, no.~8, pp. 741--745, 2002.

\bibitem{PaAgBu15}
R.~Patel, P.~Agharkar, and F.~Bullo, ``Robotic surveillance and {M}arkov chains
  with minimal weighted {K}emeny constant,'' \emph{IEEE Trans. Autom. Control},
  vol.~60, no.~12, pp. 3156--3167, Dec. 2015.

\bibitem{JaOl19}
A.~Jadbabaie and A.~Olshevsky, ``Scaling laws for consensus protocols subject
  to noise,'' \emph{IEEE Trans. Autom. Control}, vol.~64, no.~4, pp.
  1389--1402, Apr. 2019.

\bibitem{XuShZhKaZh20}
W.~Xu, Y.~Sheng, Z.~Zhang, H.~Kan, and Z.~Zhang, ``Power-law graphs have
  minimal scaling of {K}emeny constant for random walks,'' in \emph{Proc. World
  Wide Web Conf.}, 2020, pp. 46--56.

\bibitem{XiZh24SIGMOD}
H.~Xia and Z.~Zhang, ``{Efficient approximation of Kemeny's constant for large
  graphs},'' \emph{Proceedings of the ACM on Management of Data}, vol.~2,
  no.~3, pp. 1--26, 2024.

\bibitem{XiZh24KDD}
------, ``{Fast computation of Kemeny's constant for directed graphs},'' in
  \emph{Proceedings of the 30th ACM SIGKDD Conference on Knowledge Discovery
  and Data Mining}, 2024, pp. 3472--3483.

\bibitem{CrDaSeVe16}
P.~Crescenzi, G.~D'angelo, L.~Severini, and Y.~Velaj, ``Greedily improving our
  own closeness centrality in a network,'' \emph{ACM Trans. Knowl. Discov.
  Data}, vol.~11, no.~1, pp. 1--32, 2016.

\bibitem{DaOlSe19}
G.~D’Angelo, M.~Olsen, and L.~Severini, ``Coverage centrality maximization in
  undirected networks,'' in \emph{Proc. AAAI Conf. Artif. Intell.}, vol.~33,
  no.~01, 2019, pp. 501--508.

\bibitem{LiPaYiZh20}
H.~Li, S.~Patterson, Y.~Yi, and Z.~Zhang, ``Maximizing the number of spanning
  trees in a connected graph,'' \emph{IEEE Trans. Inf. Theory}, vol.~66, no.~2,
  pp. 1248--1260, 2020.

\bibitem{FrGaGuMa15}
F.~Frati, S.~Gaspers, J.~Gudmundsson, and L.~Mathieson, ``Augmenting graphs to
  minimize the diameter,'' \emph{Algorithmica}, vol.~72, pp. 995--1010, 2015.

\bibitem{AdGi22}
F.~Adriaens and A.~Gionis, ``Diameter minimization by shortcutting with degree
  constraints,'' in \emph{Proc. 2022 IEEE Int. Conf. Data Min.}\hskip 1em plus
  0.5em minus 0.4em\relax IEEE, 2022, pp. 843--848.

\bibitem{ClMoStWi79}
A.~K. Cline, C.~B. Moler, G.~W. Stewart, and J.~H. Wilkinson, ``An estimate for
  the condition number of a matrix,'' \emph{SIAM Journal on Numerical
  Analysis}, vol.~16, no.~2, pp. 368--375, 1979.

\bibitem{SpTe11}
D.~A. Spielman and S.-H. Teng, ``Spectral sparsification of graphs,''
  \emph{SIAM J. Comput.}, vol.~40, no.~4, pp. 981--1025, 2011.

\end{thebibliography}

\begin{IEEEbiography}
[{\includegraphics[width=1in,height=1.25in,clip,keepaspectratio]{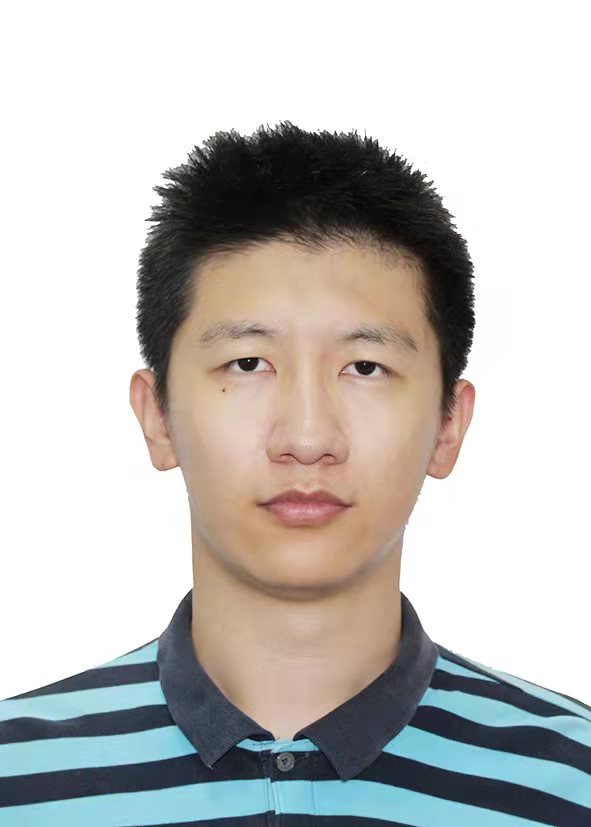}}]
{Xiaotian Zhou}
received the B.S. degree in mathematics science from Fudan University, Shanghai, China, in 2020. He is currently pursuing the PhD's degree in School of Computer Science, Fudan University, Shanghai, China. His research interests include network science, computational social science, graph data mining, and social network analysis.
\end{IEEEbiography}

\begin{IEEEbiography}
[{\includegraphics[width=1in,height=1.25in,clip,keepaspectratio]{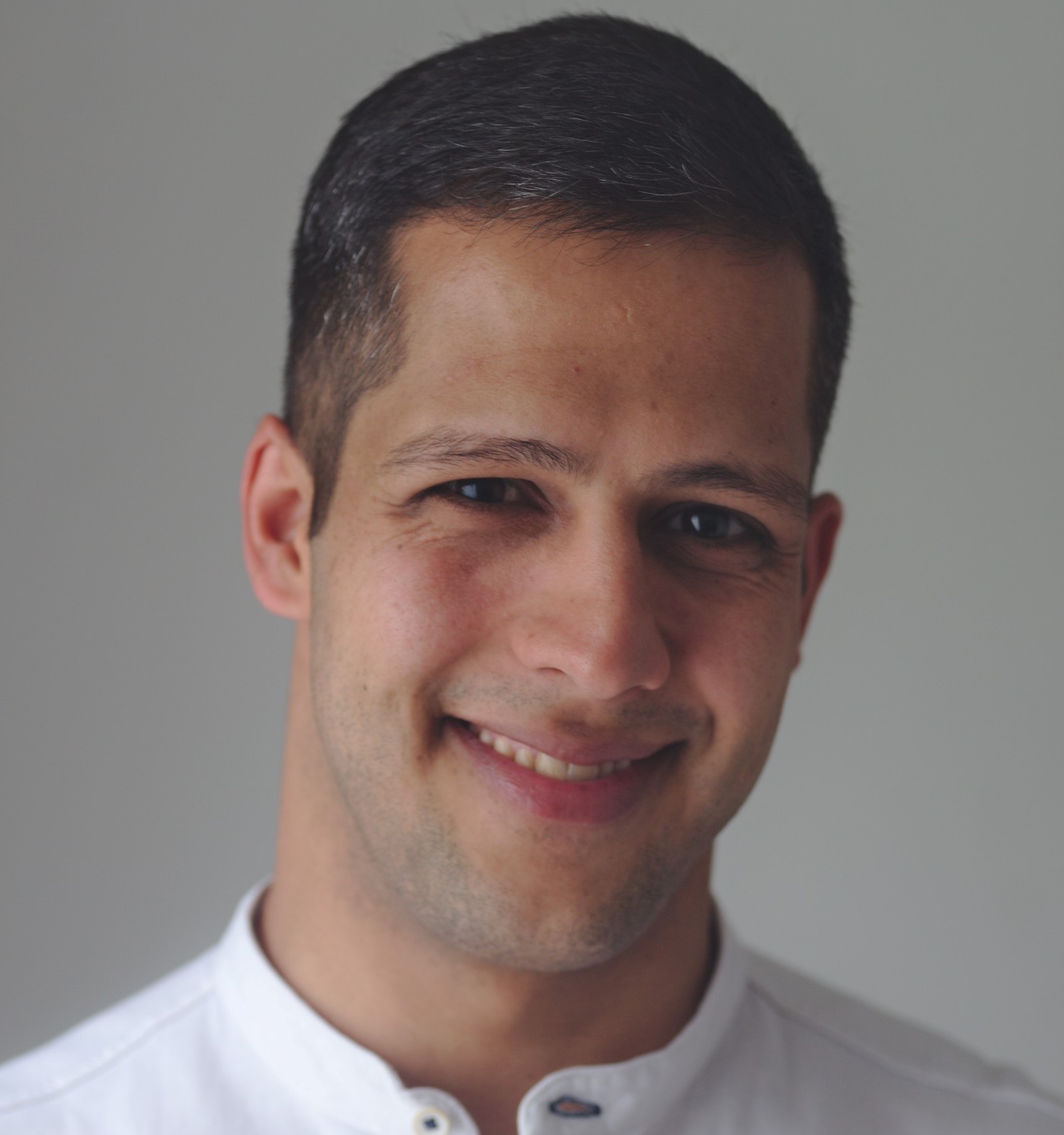}}]
{Ahad N. Zehmakan} received the PhD degree from ETH Zurich, in 2020, and is currently an assistant professor of computer science with the School of Computing, Australian National University. His research interests include graph and randomized algorithms, information spreading in social networks, random graph models, complexity theory, parallel and distributed computing, and network security.
\end{IEEEbiography}

\begin{IEEEbiography}
[{\includegraphics[width=1in,height=1.25in,clip,keepaspectratio]{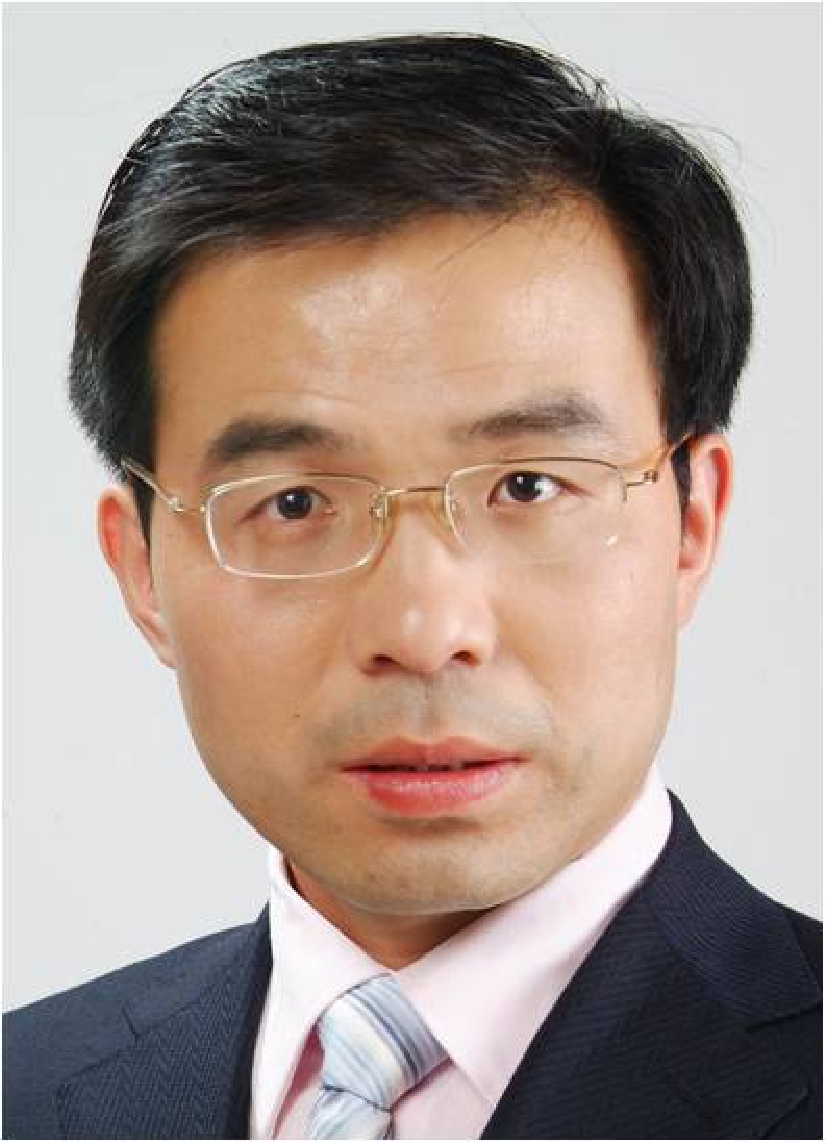}}]
{Zhongzhi Zhang} (Member, IEEE)
received the BSc degree in applied mathematics from Anhui University, Hefei, China, in 1997, and the PhD degree
in management science and engineering from the
Dalian University of Technology, Dalian, China, in
2006. From 2006 to 2008, he was a post-doctoral
research fellow with Fudan University, Shanghai,
China, where he is currently a full professor with
the School of Computer Science. He has published
more than 160 papers in international journals or
conferences. Since 2019, he has been selected as one
of the most cited Chinese researchers (Elsevier) every year. His current research
interests include network science, graph data mining, social network analysis,
computational social science, spectral graph theory, and random walks. He was
a recipient of the Excellent Doctoral Dissertation Award of Liaoning Province,
China, in 2007, the Excellent Post-Doctor Award of Fudan University, in 2008,
the Shanghai Natural Science Award (third class), in 2013, the CCF Natural
Science Award (second class), in 2022, and the Wilkes Award for the best paper
published in The Computer Journal, in 2019.
\end{IEEEbiography}

\end{document}